\documentclass[natbib,numbers,preprint]{sigplanconf}

\usepackage{amsmath,amssymb,amsthm}
\usepackage{microtype}
\usepackage{array}
\usepackage{tikz}
\usepackage{clrscode3e}
\usepackage{multirow,bigdelim}
\usepackage{calc}
\usepackage{xcolor}
\usepackage{hyperref,float}
\usepackage{mathtools}
\usepackage{enumitem}

\usepackage{siunitx}

\usetikzlibrary{arrows,automata,fit,positioning}

\theoremstyle{plain}
\newtheorem{theorem}{Theorem}
\newtheorem*{theorem*}{Theorem}
\newtheorem{lemma}{Lemma}
\newtheorem*{lemma*}{Lemma}
\newtheorem{corollary}{Corollary}

\theoremstyle{definition}
\newtheorem{definition}{Definition}

\newcommand{\Union}{\bigcup}
\newcommand{\EAlph}{\ensuremath{\mathbb{A}}} 
\newcommand{\N}{\mathbb{N}} 
\newcommand{\restr}[2]{{#1}\!\!\restriction_{#2}} 
\newcommand{\supp}{\mathtt{supp}} 
\newcommand{\orb}{\mathtt{orb}} 

\newcommand{\lang}{\mathcal{L}}

\newcommand{\pref}{\mathtt{prefixes}}
\newcommand{\suff}{\mathtt{suffixes}}

\newcommand\Lext{{\cdot}}

\newcommand{\algstep}[1]{\textbf{Step} $#1$}

\newcommand{\autom}{\mathcal{A}}

\newcommand{\lder}[1]{#1^{-1}}

\newcommand{\rowincl}{\sqsubseteq}
\newcommand{\rowincls}{\sqsubset}
\newcommand{\rowunion}{\sqcup}
\newcommand{\Rowunion}{\bigsqcup}

\newcommand{\LStar}{\ensuremath{\mathtt{L}^\star}}
\newcommand{\NLStar}{\ensuremath{\mathtt{NL}^\star}}
\newcommand{\nLStar}{\ensuremath{\nu\LStar}}
\newcommand{\nNLStar}{\ensuremath{\nu\NLStar}}

\newcolumntype{C}[1]{>{\centering\arraybackslash$}m{#1}<{$}}
\newcolumntype{L}[1]{>{\raggedleft\arraybackslash$}m{#1}<{$}}
\newcolumntype{R}[1]{>{\raggedright\arraybackslash$}m{#1}<{$}}
\tikzset{trlab/.style={font=\scriptsize,inner sep=1pt,outer sep=1pt}}
\newenvironment{automaton}[1][]{
\begin{tikzpicture}[
state/.append style={inner sep=0pt,outer sep=0pt,minimum size=3.5ex},%
initial text={},->,>=stealth',shorten >=1pt,node distance=9ex,semithick,#1]
}
{\end{tikzpicture}
}
\renewcommand{\gets}{\leftarrow}
\newcommand{\lmark}[1]{%
\tikz[overlay,remember picture] 
 \node (marker-#1-a) at (0,6pt) {};%
}
\newcommand{\rmark}[1]{%
\tikz[overlay,remember picture] 
	\node (marker-#1-b) at (0,-1pt) {};%
\tikz[overlay,remember picture,inner sep=1pt]
	\node[draw,rectangle,densely dotted,fit=(marker-#1-a.center) (marker-#1-b.center)] (#1) {};%
}

\title{Learning Nominal Automata\thanks{Work partially supported by the Polish National Science Centre (NCN) grant 2012/07/E/ST6/03026, the ERC starting grant Profoundnet (679127) and a Leverhulme Prize (PLP- 2016-129).}}

\authorinfo{Joshua Moerman}{Radboud University, The Netherlands}{joshua.moerman@cs.ru.nl}
\authorinfo{Matteo Sammartino \\ Alexandra Silva}{University College London, UK}{\{m.sammartino,alexandra.silva\}@ucl.ac.uk}
\authorinfo{Bartek Klin \\ Micha\l\ Szynwelski}{University of Warsaw, Poland}{\{klin,szynwelski\}@mimuw.edu.pl}

\preprintfooter{Paper accepted at POPL 2017 with DOI http://dx.doi.org/10.1145/3009837.3009879}

\begin{document}
\toappear{}
\maketitle


\begin{abstract}
We present an Angluin-style algorithm to learn nominal automata, which are acceptors of languages over infinite (structured) alphabets.
The abstract approach we take allows us to seamlessly extend known variations of the algorithm to this new setting. In particular we can learn a subclass of nominal non-deterministic automata. An implementation using a recently developed Haskell library for nominal computation is provided for preliminary experiments. 
\end{abstract}



\section{Introduction}
Automata are a well established computational abstraction with a wide range of applications, including modelling and verification of (security) protocols, hardware, and software systems. In an ideal world, a model would be available before a system or protocol is deployed in order to provide ample opportunity for checking important properties that must hold and only then the actual system would be synthesized from the verified model. Unfortunately, this is not at all the reality: systems and protocols are developed and coded in short spans of time and if mistakes occur they are most likely found after deployment. In this context, it has become popular to infer or learn a model from a given system just by observing its behaviour or response to certain queries. The learned model can then be used to ensure the system is complying to desired properties or to detect bugs and design possible fixes. 

Automata learning, or regular inference \cite{Angluin87}, is a widely used technique for creating an automaton model from observations. The original algorithm \cite{Angluin87}, by Dana Angluin, works for deterministic finite automata, but since then has been extended to other types of automata \cite{AngluinC97,Aarts10,Niese03}, including Mealy machines and I/O automata, and even a special class of context-free grammars. Angluin's algorithm is sometimes referred to as {\em active learning}, because it is based on direct interaction of the learner with an oracle (``the Teacher'') that can answer different types of queries. This is in contrast with \emph{passive} learning, where a fixed set of positive and negative examples is given and no interaction with the system is possible.

In this paper, staying in the realm of active learning, we will extend Angluin's algorithm to a richer class of automata. We are motivated by situations in which a program model, besides control flow, needs to represent basic data flow, where data items are compared for equality (or for other theories such as total ordering). In these situations, values for individual symbols are typically drawn from an infinite domain and automata over~\emph{infinite alphabets} become natural models, as witnessed by a recent trend~\cite{Tomte15,DAntoni14,BojanczykKL14,BolligHLM13,Cassel16}.

One of the foundational approaches to formal language theory for infinite alphabets uses the notion of nominal sets~\cite{BojanczykKL14}. The theory of nominal sets originates from the work of Fraenkel in 1922, and they were originally used to prove the independence of the axiom of choice and other axioms. They have been rediscovered in Computer Science by Gabbay and Pitts~\cite{Pitts13}, as an elegant formalism for modeling name binding, and since then they form the basis of many research projects in the semantics and concurrency community.
In a nutshell, nominal sets are infinite sets equipped with symmetries which make them finitely representable and tractable for algorithms. We make crucial use of this feature in the development of a learning algorithm. 

Our main contributions are the following.
\begin{itemize}
	\item A generalization of Angluin's original algorithm to nominal automata. The generalization follows a generic pattern for transporting computation models from finite sets to nominal sets, which leads to simple correctness proofs and opens the door to further generalizations. The use of nominal sets with different symmetries also creates potential for generalization, e.g.~to languages with time features~\cite{BL12} or data dependencies represented as graphs~\cite{MontanariS14}. 
	\item An extension of the algorithm to nominal non-deterministic automata (nominal NFAs). To the best of our knowledge, this is the first learning algorithm for non-deterministic automata over infinite alphabets.
	It is important to note that, in the nominal setting, NFAs are strictly more expressive than DFAs.
	We learn a subclass of the languages accepted by nominal NFAs, which includes all the languages accepted by nominal DFAs.
	The main advantage of learning NFAs directly is that they can provide exponentially smaller automata when compared to their deterministic counterpart.
	This can be seen both as a generalization and as an optimization of the algorithm.
	\item An implementation using our recently developed Haskell library tailored to nominal computation -- NLambda~\cite{nlambda-paper}. Our implementation is the first non-trivial application of a novel programming paradigm of functional programming over infinite structures, which allows the programmer to rely on convenient intuitions of searching through infinite sets in finite time. 
\end{itemize}
The paper is organized as follows. In Section~\ref{sec:overview}, we present an overview of our contributions (and the original algorithm) highlighting the challenges we faced in the various steps. In Section~\ref{sec:prelim}, we revise some basic concepts of nominal sets and automata. Section~\ref{sec:nangluin} contains the core technical contributions of our paper: the new algorithm and proof of correctness.  In Section~\ref{sec:nondet}, we describe an algorithm to learn nominal non-deterministic automata. Section~\ref{sec:impl} contains a description of NLambda, details of the implementation, and results of preliminary experiments.
Section~\ref{sec:related} contains a discussion of related work. We conclude the paper with a discussion section where also future directions are presented. 

\begin{figure}[t]
\begin{codebox}
\Procname{$\proc{\LStar\ learner}$}
\li $S,E \gets \{\epsilon\}$ 
\li \Repeat
\li \While $(S, E)$ is not closed or not consistent\label{line:checks}
\li \If $(S, E)$ is not closed
\li \Then\label{line:begin-closed}
find $s_1\in S$, $a \in A$ such that
\zi \qquad $row(s_1a) \neq row(s)$, for all $s\in S$
\label{line:clos-witness}
\li $S\gets S\cup \{s_1a\}$ \label{line:addrow-clos}
\End\label{line:end-closed}
\li \If $(S, E)$ is not consistent\label{line:begin-const}
\li \Then find $s_1, s_2 \in S$, $a \in A$, and $e\in E$ such that
\zi \qquad $row(s_1) = row(s_2)$ and $\lang(s_1 a e) \neq \lang(s_2 a e)$
\label{line:cons-witness}
\li $E\gets E\cup \{ae\}$
\End
\label{line:end-const}
\li Make the conjecture $M(S,E)$
\label{line:conj}
\li \If the Teacher replies  {\bf no}, with a counter-example $t$\label{line:counter-ex}
\li \Then $S\gets S\cup \pref(t)$ \label{line:addrow-ex}
\End
\li \Until the Teacher replies {\bf yes} to the conjecture $M(S,E)$.
\li \Return $M(S,E)$
\end{codebox}
\caption{Angluin's algorithm for deterministic finite automata~\cite{Angluin87}}\label{fig:alg}
\end{figure}

\section{Overview of the approach}\label{sec:overview}

In this section, we give an overview of the work developed in the paper through examples. We will start by explaining the original algorithm for regular languages over finite alphabets, and then explain the challenges in extending it to nominal languages. 

Angluin's algorithm \LStar\ provides a procedure to learn the minimal DFA accepting a certain (unknown) language $\mathcal L$.
The algorithm has access to a \emph{teacher} which answers two types of queries:
\begin{itemize}
	\item \emph{membership queries}, consisting of a single word $w \in A^\star$, to which the teacher will reply whether $w \in \lang$ or not;
	\item \emph{equivalence queries}, consisting of a hypothesis DFA $H$, to which the teacher replies \textbf{yes} if $\lang(H) = \lang$, and \textbf{no} otherwise, providing a counterexample $w \in \lang(H) \triangle \lang$ ($\triangle$ denotes the symmetric difference of two languages).
\end{itemize}
The learning algorithm works by incrementally building an {\em observation table}, which at each stage contains partial information about the language $\lang$.
The algorithm is able to fill the table with membership queries.
As an example, and to set notation, consider the following table (over $A=\{a,b\}$). 
\begin{center}
\includegraphics{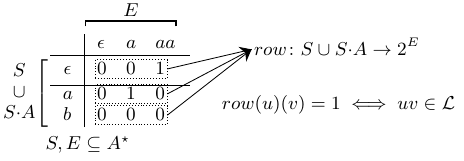}
\end{center}
This table indicates that $\mathcal L$ contains at least $aa$ and definitely does not contain the words $\epsilon, a, b, ba, baa, aaa$. Since $row$ is fully determined
by the language $\lang$, we will from now on refer to an
observation table as a pair $(S,E)$, leaving the language
$\lang$ implicit.

Given an observation table $(S,E)$ one can construct a deterministic automaton $M(S,E) = (Q,q_0,\delta,F)$ where
\begin{itemize}
\item $Q 
 = 
\{row(s)\mid s\in S\}$ is a finite set of states;
\item $F 
 =  
\{ row (s) \mid s\in S, row(s)(\epsilon)=1\} \subseteq Q$ is the set of final states;
\item  $q_0 = row(\epsilon)$ is the initial state;
\item $\delta\colon Q\times A \to Q$ is the transition function given by $\delta(row(s),a)
 = 
row(sa)$.
\end{itemize}
For this to be well-defined, we need to have $\epsilon \in S$ (for the initial state) and $\epsilon \in E$ (for final states), and for the transition function there are two crucial properties of the table that need to hold: 
closedness and consistency. An observation table \((S,E)\) is {\em closed} if for all $t\in S \Lext A$ there exists an $s \in S$ such that $row(t) = row(s)$.  An
observation table \((S,E)\) is {\em consistent} if, whenever $s_1$ and
$s_2$ are elements of $S$ such that $row(s_1) = row(s_2)$, for all $a
\in A$, $row(s_1 a) = row(s_2 a)$. 
Each time the algorithm constructs an automaton, it poses an equivalence query to the teacher.
It terminates when the answer is \textbf{yes}, otherwise it extends the table with the counterexample provided.

\subsection{Simple example of execution}
\label{sec:execution_example_original}
Angluin's algorithm is displayed in Figure~\ref{fig:alg}. Throughout this section, we will consider the language(s)
\[
	\lang_n = \{ ww \mid w \in A^\star, |w| = n \}
\]
If the alphabet $A$ is finite then $\lang_n$ is regular for any $n \in \N$, and there is a finite DFA accepting it. 

The language $\lang_1 = \{ aa , bb \}$ looks trivial, but the minimal DFA recognizing it has as many as 5 states. Angluin's algorithm will terminate in (at most) 5 steps. We illustrate some relevant ones.

\paragraph{\algstep{1}.}
We start from $S,E=\{\epsilon\}$, and we fill the entries of the table below by asking membership queries for $\epsilon$, $a$ and $b$. The table is closed and consistent, so we construct the hypothesis $\autom_1$.
\begin{center}
\begin{tabular}{m{.4\linewidth}m{.4\linewidth}}
\begin{tabular}{L{2em} L{1ex} | C{1ex} @{}m{0pt}@{}}
& & \epsilon &\\
\cline{2-3}
& \epsilon & 0 &\\[1ex]
\cline{2-3}
& a & 0 &\\[.5ex]
& b & 0 &
\end{tabular}
&
\begin{gather*}
\autom_1 = 	
\begin{automaton}[baseline=-.5ex]
\node[initial,state] (q0) {$q_0$};
\path (q0) edge[loop right] node[trlab,right] {$a,b$} (q0); 	
\end{automaton}	
\\
q_0 = row(\epsilon) = \{ \epsilon \mapsto 0 \}
\end{gather*}
\end{tabular}
\end{center}
The Teacher replies \textbf{no} and gives the counterexample $aa$, which is in $\lang_1$ but it is not accepted by $\autom_1$. Therefore, line 12 of the algorithm is triggered and we set $S \gets S \cup \{a,aa\}$.
\paragraph{\algstep{2}.}
\label{alg:step2}
The table becomes the one on the left below. It is closed, but not consistent: rows $\epsilon$ and $a$ are identical, but appending $a$ leads to different rows, as depicted. Therefore, line 9 is triggered and an extra column $a$, highlighted in red, is added. The new table is closed and consistent and a new hypothesis $\autom_2$ is constructed.
\begin{center}
\begin{tabular}{m{.2\linewidth}m{.2\linewidth}m{.4\linewidth}}
\centering	
\begin{tabular}{R{3ex}|C{1ex}}
& \epsilon \\
\hline
\lmark{r1} \epsilon & 0 \rmark{r1} \\
\lmark{r2} a & 0 \rmark{r2} \\
\lmark{r3} aa & 1 \rmark{r3} \\
\hline
b & 0 \\
ab & 0 \\
aaa & 0 \\
aab & 0 
\end{tabular}
&
\centering
\begin{tabular}{R{3ex}|C{.5ex}>{\color{red}}C{.5ex}}
& \epsilon & a \\
\hline
\epsilon & 0 & 0\\
a & 0 & 1 \\
aa & 1 & 0 \\
\hline
b & 0  & 0\\
ab & 0 & 0 \\
aaa & 0 & 0 \\
aab & 0 & 0
\end{tabular}
&
\begin{gather*}
\autom_2 = 
\hspace{-10pt}
\begin{gathered}
\begin{automaton}
\node[initial,state] (q0) {$q_0$};	
\node[state,right of=q0] (q1) {$q_1$};	
\node[state,accepting,below of=q1] (q2) {$q_2$};	
\path 
(q0) edge[loop below] node[trlab,below] {$b$} (q0)
(q0) edge[bend right] node[trlab,above] {$a$} (q1)
(q1) edge[bend right] node[trlab,above] {$b$} (q0)
(q1) edge node[trlab,right] {$a$} (q2)
(q2) edge node[trlab,midway,fill=white] {$a, b$} (q0);
\end{automaton}
\end{gathered}	
\end{gather*}
\begin{tikzpicture}[overlay,remember picture,->,>=stealth',shorten >=1pt]
	\path (r1.west) edge[in=160,out=190,looseness=2] node[trlab,left] {$a$} (r2.west);		
	\path (r2.east) edge[in=20,out=-10,looseness=2] node[trlab,right] {$a$} (r3.east);		
\end{tikzpicture}
\end{tabular}	
\end{center}
The Teacher again replies \textbf{no} and gives the counterexample $bb$, which should be accepted by $\autom_2$ but it is not. Therefore we put $S \gets S \cup \{b,bb\}$.

\paragraph{\algstep{3}.}
\label{alg:step3}
The new table is the one on the left. It is closed, but $\epsilon$ and $b$ violate consistency, when $b$ is appended. Therefore we add the column $b$ and we get the table on the right, which is closed and consistent. The new hypothesis is $\autom_3$.
\begin{flushleft}
\begin{tabular}{m{.25\linewidth}m{.25\linewidth}m{.2\linewidth}}
\raggedright
\begin{tabular}{R{3ex}|C{.5ex}C{.5ex}}
& \epsilon & a \\
\hline
\lmark{r4} \epsilon & 0 & 0 \rmark{r4} \\
a & 0 & 1 \\
aa & 1 & 0 \\
\lmark{r5} b & 0 & 0 \rmark{r5} \\
\lmark{r6} bb & 1 & 0 \rmark{r6} \\
\hline
ab & 0 & 0 \\
aaa & 0 & 0 \\
aab & 0 & 0 \\
ba & 0 & 0 \\
bba & 0 & 0 \\
bbb & 0 & 0 \\
\end{tabular}	
&
\centering
\begin{tabular}{R{3ex}|C{.5ex}C{.5ex}>{\color{red}}C{.5ex}}
& \epsilon & a & b \\
\hline
\epsilon & 0 & 0 & 0 \\
a & 0 & 1 & 0 \\
aa & 1 & 0 & 0 \\
b & 0 & 0 & 1\\
bb & 1 & 0 & 0 \\
\hline
ab & 0 & 0 & 0 \\
aaa & 0 & 0 & 0 \\
aab & 0 & 0 & 0 \\
ba & 0 & 0 & 0 \\
bba & 0 & 0 & 0 \\
bbb & 0 & 0 & 0 \\
\end{tabular}	
&
\raggedright
\begin{gather*}
\autom_3 = 
\hspace{-10pt}
\begin{gathered}
\begin{automaton}
\node[initial,state] (q0) {$q_0$};	
\node[state,right of=q0] (q1) {$q_1$};	
\node[state,below of=q0] (q2) {$q_2$};	
\node[state,accepting,right of=q2] (q3) {$q_3$};	
\path 
(q0) edge[bend left] node[trlab,above] {$a$} (q1)
(q1) edge[bend left] node[trlab,above] {$b$} (q0)
(q1) edge node[trlab,right] {$a$} (q3)
(q0) edge[bend right] node[trlab,left] {$b$} (q2)
(q2) edge[bend right] node[trlab,left] {$a$} (q0)
(q2) edge node[trlab,above] {$b$} (q3)
(q3) edge node[trlab,fill=white,midway] {$a,b$} (q0);
\end{automaton}
\end{gathered}
\end{gather*}
\begin{tikzpicture}[overlay,remember picture,->,>=stealth',shorten >=1pt]
	\path (r4.west) edge[in=160,out=190,looseness=1] node[trlab,left] {$b$} (r5.west);		
	\path (r5.east) edge[in=20,out=-10,looseness=2] node[trlab,right] {$b$} (r6.east);		
\end{tikzpicture}
\end{tabular}	
\end{flushleft}
The Teacher replies \textbf{no} and provides the counterexample $babb$, so $S \gets S \cup \{ba,bab\}$.

\paragraph{\algstep{4}.}
One more step brings us to the correct hypothesis $\autom_4$ (details are omitted).
\[
\autom_4 = 
\begin{gathered}
\begin{automaton}
\node[initial,state] (q0) {$q_0$};	
\node[state,above right= 15pt of q0] (q1) {$q_1$};	
\node[state,below right= 15pt of q0] (q2) {$q_2$};	
\node[state,accepting,right= 30pt of q0] (q3) {$q_3$};	
\node[state,right of=q3] (q4) {$q_4$};	
\path 
(q0) edge node[trlab,left] {$a$} (q1)
(q1) edge node[trlab,above] {$a$} (q3)
(q0) edge node[trlab,left] {$b$} (q2)
(q2) edge node[trlab,below] {$b$} (q3)
(q1) edge[bend left] node[trlab,above] {$b$} (q4)
(q3) edge node[trlab,above] {$a,b$} (q4)
(q2) edge[bend right] node[trlab,below] {$a$} (q4)
(q4) edge[loop right] node[trlab,right] {$a,b$} (q4);
\end{automaton}
\end{gathered}
\]
\subsection{Learning nominal languages}
\label{ssec:nom-learning}
Consider now an infinite alphabet $A = \{a,b,c,d,\dots\}$. The language $\lang_1$ becomes $\{aa,bb,cc,dd,\dots\}$. Classical theory of finite automata does not apply to this kind of languages, but one may draw an infinite deterministic automaton that recognizes $\lang_1$ in the standard sense:
\[
\autom_5 = 
\begin{gathered}
\begin{automaton}
\node[initial,state] (q0) {$q_0$};	
\node[state,above right= 25pt of q0] (qa) {$q_a$};	
\node[state, right= 15pt of q0] (qb) {$q_b$};	
\node[state,accepting,right= 15pt of qb] (q3) {$q_3$};	
\node[state,right of=q3] (q4) {$q_4$};	
\node[below right = 25pt of q0] (qdots) {$\vdots$};
\path 
(q0) edge node[trlab,left] {$a$} (qa)
(qa) edge node[trlab,above] {$a$} (q3)
(q0) edge node[trlab,below] {$b$} (qb)
(qb) edge node[trlab,above] {$b$} (q3)
(q0) edge node[trlab] {} (qdots)
(qdots) edge node[trlab] {} (q3)
(qa) edge[bend left] node[trlab,above] {$\neq a$} (q4)
(q3) edge node[trlab,above] {$A$} (q4)
(qb) edge[bend right] node[trlab,below] {$\neq b$} (q4)
(qdots) edge[bend right] node[trlab] {} (q4)
(q4) edge[loop right] node[trlab,right] {$A$} (q4);
\end{automaton}
\end{gathered}
\]
where $\xrightarrow{A}$ and $\xrightarrow{\neq a}$ stand for the infinitely-many transitions labelled by elements of $A$ and $A \setminus \{a\}$, respectively.
This automaton is infinite, but it can be finitely presented in a variety of ways, for example:
\begin{align}\label{eq:aut}
\begin{gathered}
\begin{automaton}
\node[initial,state] (q0) {$q_0$};	
\node[state, right= 15pt of q0] (qx) {$q_x$};	
\node[state,accepting,right= 15pt of qx] (q3) {$q_3$};	
\node[state,right of=q3] (q4) {$q_4$};	
\node[above = 0pt of qx] (qa) {\scriptsize{$\forall x\!\in\! A$}};
\path 
(q0) edge node[trlab,above] {$x$} (qx)
(qx) edge node[trlab,above] {$x$} (q3)
(q3) edge node[trlab,above] {$A$} (q4)
(qb) edge[bend right] node[trlab,below] {$\neq x$} (q4)
(q4) edge[loop right] node[trlab,right] {$A$} (q4);
\end{automaton}
\end{gathered}
\end{align}
One can formalize the quantifier notation above (or indeed the ``dots'' notation above that) in  several ways. A popular solution is to consider finite {\em register automata}~\cite{kamFranc,DemriL09}, i.e., finite automata equipped with a finite number of registers where alphabet letters can be stored and later compared for equality. Our language $\lang_1$ is recognized by a simple automaton with four states and one register. The problem of learning registered automata has been successfully attacked before~\cite{HowarSJC12}.

In this paper, however, we will consider nominal automata~\cite{BojanczykKL14} instead. These automata ostensibly have infinitely many states, but the set of states can be finitely presented in a way open to effective manipulation. More specifically, in a nominal automaton the set of states is subject to an action of permutations of a set $\EAlph$ of {\em atoms}, and it is finite up to that action. For example, the set of states of $\autom_5$ is:
\[
	\{q_0,q_3,q_4\} \cup \{q_a\mid a\in A\}
\]
and it is equipped with a canonical action of permutations $\pi \colon \EAlph \to \EAlph$ that maps every $q_a$ to $q_{\pi_a}$, and leaves $q_0$, $q_3$ and $q_4$ fixed. Technically speaking, the set of states has four {\em orbits} (one infinite orbit and three fixed points) of the action of the group of permutations of $\EAlph$.
Moreover, it is required that in a nominal automaton the transition relation is {\em equivariant}, i.e., closed under the action of permutations. The automaton $\autom_5$ has this property: for example, it has a transition $q_a\stackrel{a}{\longrightarrow} q_3$, and for any $\pi \colon \EAlph \to \EAlph$ there is also a transition $\pi(q_a)=q_{\pi(a)}\stackrel{\pi(a)}{\longrightarrow} q_3=\pi(q_3)$.

Nominal automata with finitely many orbits of states are equi-expressive with finite register automata~\cite{BojanczykKL14}, but they have an important theoretical advantage: they are a direct reformulation of the classical notion of finite automaton, where one replaces finite sets with orbit-finite sets and functions (or relations) with equivariant ones. A research programme advocated in~\cite{BojanczykKL14,BBKL12} is to transport various computation models, algorithms and theorems along this correspondence. This can often be done with remarkable accuracy, and our paper is a witness to this. Indeed, as we shall see, nominal automata can be learned with an algorithm that is almost a verbatim copy of the classical Angluin's one.

Indeed, consider applying Angluin's algorithm to our new language $\lang_1$. The key idea is to change the basic data structure: our observation table $(S,E)$ will be such that \emph{$S$ and $E$ are equivariant subsets of $A^\star$}, i.e., they are closed under the canonical action of atom permutations. In general, such a table has {\em infinitely many rows and columns}, so the following aspects of the algorithm seem problematic:
\begin{description}
	\item[line \ref{line:checks}:] closedness and consistency tests range over infinite sets;
	\item[line \ref{line:clos-witness} and \ref{line:cons-witness}:] finding witnesses for closedness or consistency violations potentially require checking all infinitely many rows;
	\item[line \ref{line:addrow-ex}:] every counterexample $t$ has only finitely many prefixes, so it is not clear how one would construct an infinite set $S$ in finite time. However, an infinite $S$ is necessary for the algorithm to ever succeed, because no finite automaton recognizes $\lang_1$. 
\end{description}
At this stage, we need to observe that due to equivariance of $S$, $E$ and $\lang_1$, the following crucial properties hold:
\begin{description}
	\item[(P1)] the sets $S$, $S \Lext A$ and $E$ admit a \emph{finite} representation up to permutations;
	\item[(P2)] the function $row$ is such that $row(\pi(s))(\pi(e)) = row(s)(e)$, for all $s \in S$ and $e \in E$, so the observation table admits a finite symbolic representation.
\end{description}
Intuitively, checking closedness and consistency, and finding a witness for their violations, can be done effectively on the representations up to permutations \textbf{(P1)}. This is sound, as $row$ is invariant w.r.t.\ permutations \textbf{(P2)}.

We now illustrate these points through a few steps of the algorithm for $\lang_1$.
\paragraph{\algstep{1'}:} We start from $S,E = \{ \epsilon \}$. We have $S \Lext A = A$, which is infinite but admits a finite representation. In fact, for any $a \in A$, we have $A = \{ \pi(a) \mid \text{$\pi$ is a permutation}\}$. Then, by \textbf{(P2)}, $row(\pi(a))(\epsilon) = row(a)(\epsilon) = 0$, for all $\pi$, so the first table can be written as:
\begin{center}
\begin{tabular}{@{\qquad\qquad}m{.3\linewidth}m{.3\linewidth}}
\begin{tabular}{L{1ex} | C{1ex}}
& \epsilon \\
\hline
\epsilon & 0 \\
\hline
a
 & 0
\end{tabular}
&
\begin{gather*}
\autom_1' = 
\begin{gathered}
\begin{automaton}
\node[initial,state] (q0) {$q_0$};
\path (q0) edge[loop right] node[trlab,right] {$A$} (q0); 	
\end{automaton}	
\end{gathered}
\end{gather*}
\end{tabular}
\end{center}
It is closed and consistent. Our hypothesis is $\autom_1'$, where $\delta_{\autom_1'}(row(\epsilon),x) = row(x) = q_0$, for all $x \in A$. As in \algstep{1}, the Teacher replies with the counterexample $aa$.
\paragraph{\algstep{2'}.} 
By equivariance of $\lang_1$, the counterexample tells us that \emph{all} words of length 2 with two repeated letters are accepted. Therefore we extend $S$ with the (infinite!) set of such words. The new symbolic table is:
\begin{center}
\begin{tabular}{R{3ex}|C{1ex}}
& \epsilon \\
\hline
\epsilon & 0 \\
a & 0 \\
aa & 1 \\
\hline
ab & 0 \\
aaa & 0 \\
aab & 0 
\end{tabular}
\end{center}
The lower part stands for elements of $S \Lext A$. For instance, $ab$ stands for words obtained by appending a fresh letter to words of length 1 (row $a$). It can be easily verified that all cases are covered. Notice that the table is different from that of \algstep{2}: a single $b$ is not in the lower part, because it can be obtained from $a$ via a permutation. The table is closed. 

Now, for consistency we need to check $row(\epsilon x) = row(ax)$, for all $a,x \in A$. Again, by \textbf{(P2)}, it is enough to consider rows of the table above. Consistency is violated, because $row(a) \neq row(aa)$. We found a ``symbolic'' witness $a$ for such violation. In order to fix consistency, while keeping $E$ equivariant, we need to add columns for all $\pi(a)$. The resulting table is 
\begin{center}
\begin{tabular}{R{3ex}|C{1ex}C{1ex}C{1ex}C{1ex}C{1ex}}
& \epsilon & a & b & c & \dots \\
\hline
\epsilon & 0 & 0 & 0 & 0 & \dots \\
a & 0 & 1 & 0 & 0 & \dots \\
aa & 1 & 0 & 0 & 0 & \dots \\
\hline
ab & 0 & 0 & 0 & 0 & \dots \\
aaa & 0 & 0 & 0 & 0 & \dots \\
aab & 0 & 0 & 0 & 0 & \dots
\end{tabular}
\end{center}
where non-specified entries are 0. Only finitely many entries of the table are relevant: $row(s)$ is fully determined by its values on letters in $s$ and on just one letter not in $s$. For instance, we have $row(a)(a) = 1$ and $row(a)(a') = 0$, for all $a' \in A \setminus \{ a \}$. The table is trivially consistent.

Notice that this step encompasses both \algstep{2} and $3$, because the rows $b$ and $bb$ added by \algstep{2} are already represented by $a$ and $aa$. The hypothesis automaton is
\[
\autom_2' = 
\begin{automaton}[baseline=-.5ex]
\node[initial,state] (q0) {$q_0$};	
\node[state,right of=q0] (qx) {$q_x$};	
\node[state,accepting,right of=qx] (q2) {$q_2$};	
\node[right of=q2] (dummy) {$\forall x \in A$};
\path 
(q0) edge[bend left] node[trlab,above] {$x$} (qx)
(qx) edge[bend left] node[trlab,above] {$\neq x$} (q0)
(qx) edge node[trlab,above] {$x$} (q2)
(q2) edge[bend right=50] node[trlab,below] {$A$} (q0);
\end{automaton}	
\]
This is again incorrect, but one additional step will give the correct hypothesis automaton, shown earlier in (\ref{eq:aut}).

\subsection{Generalization to non-deterministic automata}
Since our extension of Angluin's \LStar\ algorithm stays close to her original development, exploring extensions of other variations of \LStar\ to the nominal setting can be done in a systematic way. We will show  how to extend the algorithm \NLStar\ for learning NFAs by Bollig et al. \cite{BolligHKL09}. This has practical implications: it is well-known that NFAs are exponentially more succinct than DFAs. This is true also in the nominal setting. However, there are challenges in the extension that require particular care.
\begin{itemize}
	\item Nominal NFAs are strictly more expressive than nominal DFAs. We will show that the nominal version of \NLStar\ terminates for all nominal NFAs that have a corresponding nominal DFA and, more surprisingly, that it is capable of learning some languages that are not accepted by nominal DFAs.  
	\item Language equivalence of nominal NFAs is undecidable. This does not affect the correctness proof, as it assumes a teacher which is able to answer equivalence queries accurately. For our implementation, we will describe heuristics that produce correct results in many cases.
\end{itemize}
For the learning algorithm the power of non-determinism means that we can make some shortcuts during learning: if we want to make the table closed, we were previously required to find an equivalent row in the upper part; now we may find a sum of rows which, together, are equivalent to an existing row. This means that in some cases fewer rows will be added for closedness.
\section{Preliminaries}
\label{sec:prelim}
We recall the notions of nominal sets, nominal automata and nominal regular languages (see \cite{BojanczykKL14} for a detailed account).

Let $\EAlph$ be a countable set and let $Perm(\EAlph)$ be the set of \emph{permutations on $\EAlph$}, i.e., the bijective functions $\pi \colon \EAlph \to \EAlph$. Permutations form a group where the identity permutation $id$ is the unit element, inverse is functional inverse and multiplication is function composition.

A \emph{nominal set}~\cite{Pitts13} is a set $X$ plus a function $\cdot \colon Perm(\EAlph) \times X \to X$, interpreting permutations over $X$. Such function must be a \emph{group action} of $Perm(\EAlph)$, i.e., it must satisfy $id \cdot x = x$ and $\pi \cdot ( \pi' \cdot x) = ( \pi \circ \pi') \cdot x$. We say that a finite $A \subseteq \EAlph$ \emph{supports} $x \in X$ whenever, for all $\pi$ acting as the identity on $A$, we have $\pi \cdot x = x$. In other words, permutations that only move elements outside $A$ do not affect $x$. The support of $x \in X$, denoted $\supp(x)$, is the smallest finite set supporting $x$.
We require nominal sets to have \emph{finite support}, meaning that $\supp(x)$ exists for all $x \in X$.

The \emph{orbit} $\orb(x)$ of $x \in X$ is the set of elements in $X$ reachable from $x$ via permutations, explicitly 
\[
\orb(x) = \{ \pi \cdot x \mid \pi \in Perm(\EAlph) \}
\]
Then $X$ is \emph{orbit-finite} whenever it is a union of finitely many orbits.

Given a nominal set $X$, a subset $Y \subseteq X$ is \emph{equivariant} if it is preserved by permutations, i.e., $\pi \cdot y \in Y$, for all $y \in Y$. In other words, $Y$ is the union of orbits of $X$. This definition extends to the notion of an equivariant relation $R \subseteq X \times Y$, by setting $\pi \cdot (x,y) = (\pi \cdot x , \pi \cdot y)$, for $(x,y) \in R$; similarly for relations of greater arity. The \emph{dimension} of nominal set $X$ is the maximal size of $\supp(x)$, for any $x \in X$. Every orbit-finite set has finite dimension. 

We define $\EAlph^{(k)} = \{ (a_1,\dots,a_k) \mid \text{$a_i \neq a_j$ for $i \neq j$}\}$. For every single-orbit nominal set $X$ with dimension $k$, there is a surjective equivariant map
\[
	f_X \colon \EAlph^{(k)} \to X \enspace .
\]
This map can be used to get an upper bound for the number of orbits of $X_1 \times \dots \times X_n$, for $X_i$ a nominal set with $l_i$ orbits and dimension $k_i$. Suppose $O_i$ is an orbit of $X_i$. Then we have a surjection
\[
	\EAlph^{(k_i)} \times \dots \times \EAlph^{(k_n)}
	\xrightarrow{f_{O_1} \times \dots \times f_{O_n}} O_1 \times \dots \times O_n 
\]
stipulating that the codomain cannot have more orbits than the domain. Let $f_\EAlph(\{k_i\})$ denote the number of orbits of $\EAlph^{(k_1)} \times \dots \times \EAlph^{(k_n)}$, for any finite sequence of natural numbers $\{k_i\}$. We can form at most $l = l_1l_2 \dots l_n$ tuples of the form $O_1 \times \dots \times O_n$, so $X_1 \times \dots \times X_n$ has at most $lf_\EAlph(k_1,\dots,k_n)$ orbits.

For $X$ single-orbit, the \emph{local symmetries} are defined by the group 
$
	\{ g \in S_k \mid f(x_1, \ldots, x_k) = f(x_{g(1)}, \ldots, x_{g(k)}) \text{ for all } x_i \in X\},
$ 
where $k$ is the dimension of $X$ and $S_k$ is the \emph{symmetric group} of permutations over $k$ distinct elements.

NFAs on sets have a finite state space. We can define \emph{nominal} NFAs, with the requirement that the state space is orbit-finite and the transition relation is equivariant. A nominal NFA is a tuple $(Q,A,Q_0,F,\delta)$, where:
\begin{itemize}
	\item $Q$ is an orbit-finite nominal set of \emph{states};
	\item $A$ is an orbit-finite nominal alphabet;
	\item $Q_0,F \subseteq Q$ are equivariant subsets of \emph{initial} and \emph{final states};
	\item $\delta \subseteq Q \times A \times Q$ is an equivariant \emph{transition relation}.
\end{itemize}
A nominal DFA is a special case of nominal NFA where $Q_0 = \{q_0\}$ and the transition relation is an equivariant function $\delta \colon Q \times A \to Q$. Equivariance here can be rephrased as requiring $\delta(\pi \cdot q, \pi \cdot a) = \pi \cdot \delta(q,a)$.
In most examples we take the alphabet to be $A = \EAlph$, but it can be any orbit-finite nominal set. For instance, $A = Act \times \EAlph$, where $Act$ is a finite set of actions, represents actions $act(x)$ with one parameter $x \in \EAlph$ (actions with arity $n$ can be represented via $n$-fold products of $\EAlph$).

A language $\lang$ is \emph{nominal regular} if it is recognized by a nominal DFA. The theory of nominal regular languages recasts the classical one using nominal concepts. A nominal Myhill-Nerode-style \emph{syntactic congruence} is defined: $w,w' \in A^\star$ are equivalent w.r.t.\ $\lang$, written $w \equiv_\lang w'$, whenever
\[
	wv \in \lang \iff w'v \in \lang
\]
for all $v \in A^\star$. This relation is equivariant and the set of equivalence classes $[w]_\lang$ is a nominal set.

\begin{theorem}[Myhill-Nerode theorem for nominal sets~\cite{BojanczykKL14}]
Let $\lang$ be a regular nominal language. The following conditions are equivalent:
\begin{enumerate}
	\item the set of equivalence classes of $\equiv_\lang$ is orbit-finite;
	\item $\lang$ is recognized by a nominal DFA.
\end{enumerate}
\end{theorem}
Unlike what happens for ordinary regular languages, nominal NFAs and nominal DFAs \emph{are not equi-expressive}. Here is an example of a language accepted by a nominal NFA, but not by a nominal DFA:
\[
	\lang_{eq} = \{ a_1 \dots a_n \mid a_i = a_j, \text{for some $i < j \in \{1,\dots,n\}$} \}
\]
In the theory of nominal regular languages, several problems are decidable: language inclusion and minimality test for nominal DFAs. Moreover, orbit-finite nominal sets can be finitely-represented, and so can be manipulated by algorithms. This is the key idea underpinning our implementation.

\subsection{Different atom symmetries}
\label{sec:other-symms}

An important advantage of nominal set theory as considered in~\cite{BojanczykKL14} is that it retains most of its properties when the structure of atoms $\EAlph$ is replaced with an arbitrary infinite relational structure subject to a few model-theoretic assumptions. An example alternative structure of atoms is the total order of rational numbers $(\mathbb{Q},<)$, with the group of monotone bijections of $\mathbb{Q}$ taking the role of the group of all permutations. The theory of nominal automata remains similar, and an example nominal language over the atoms $(\mathbb{Q},<)$ is:
\[
	 \{ a_1 \dots a_n \mid a_i \leq a_j, \text{for some $i < j \in \{1,\dots,n\}$} \}
\]
which is recognized by a nominal DFA over those atoms.

To simplify the presentation, in this paper we concentrate on the ``equality atoms'' only. Also our implementation of nominal learning algorithms is restricted to equality atoms. However, both the theory and the implementation can be generalized to other atom structures, with the ``ordered atoms'' $(\mathbb{Q},<)$ as the simplest other example. We leave the details of this for a future extended version of this paper.
\section{Angluin's algorithm for nominal DFAs}
\label{sec:nangluin}

In our algorithm, we will assume a teacher as described at the start of Section~\ref{sec:overview}. In particular, the teacher is able to answer membership queries and equivalence queries, now in the setting of  nominal languages.
We fix a target language $\lang$, which is assumed to be a nominal regular language.

The learning algorithm for nominal automata, $\nLStar$, will be very similar to $\LStar$ in Figure~\ref{fig:alg}. In fact, we only change the following lines:
\begin{gather}\label{eq:changes}
\begin{array}{ll}
	6' &S \gets S \cup \orb(sa) \\
	9' &E \gets E \cup \orb(ae) \\
	12' &S \gets S \cup \pref(\orb(t))
\end{array}
\end{gather}
The basic data structure is an \emph{observation table} $(S,E,T)$ where $S$ and $E$ are orbit-finite subsets of $A^\star$ and $T : S \cup S \Lext A \times E \to 2$ is an equivariant function defined by $T(se) = \lang(se)$ for each $s \in S \cup S \Lext A$ and $e \in E$.
Since $T$ is determined by $\lang$ we omit it from the notation.
Let $row : S \cup S \Lext A \to 2^E$ denote the curried counterpart of $T$. Let $u \sim v$ denote the relation $row(u) = row(v)$.

\begin{definition}
	The table is called \emph{closed} if for each $t \in S \Lext A$ there is a $s \in S$ with $t \sim s$.
	The table is called \emph{consistent} if for each pair $s_1, s_2 \in S$ with $s_1 \sim s_2$ we have $s_1 a \sim s_2 a$ for all $a \in A$.
\end{definition}

The above definitions agree with the abstract definitions given in \cite{JacobsSilva14} and we may use some of their results implicitly.
The intuition behind the definitions is as follows. Closedness assures us that for each state we have a successor state for each input. Consistency assures us that each state has at most one successor for each input.
Together it allows us to construct a well-defined minimal automaton from the observations in the table.

The algorithm starts with a trivial observation table and tries to make it closed and consistent by adding orbits of rows and columns, filling the table via membership queries.
When the table is closed and consistent it constructs a hypothesis automaton and poses an equivalence query.

The pseudocode for the nominal version is the same as listed in Figure~\ref{fig:alg}, modulo the changes displayed in (\ref{eq:changes}).
However, we have to take care to ensure that all manipulations and tests on the (possibly) infinite sets $S, E$ and $A$ terminate in finite time.
We refer to \cite{BojanczykKL14} and \cite{Pitts13} for the full details on how to represent these structures and provide a brief sketch here.
The sets $S, E, A$ and $S \Lext A$ can be represented by choosing a representative for each orbit.
The function $T$ in turn can be represented by cells $T_{i,j} : \orb(s_i) \times \orb(e_j) \to 2$ for each representative $s_i$ and $e_j$.
Note, however, that the product of two orbits may consist of several orbits, so that $T_{i,j}$ is not a single boolean value.
Each cell is still orbit-finite and can be filled with only finitely many membership queries.
Similarly the curried function $row$ can be represented by a finite structure.

To check whether the table is closed, we observe that if we have a corresponding row $s \in S$ for some $t \in S \Lext A$, this holds for any permutation of $t$.
Hence it is enough to check the following: for all representatives $t \in S \Lext A$ there is a representative $s \in S$ with $row(t) = \pi \cdot row(s)$ for some permutation $\pi$. Note that we only have to consider finitely many permutations, since the support is finite and so we can decide this property. Furthermore if the property does not hold, we immediately find a witness represented by $t$.

Consistency is a bit more complicated, but it is enough to consider the set of inconsistencies, $\{(s_1, s_2, a, e) \mid row(s_1) = row(s_2) \wedge row(s_1 a)(e) \neq row(s_2 a)(e)\}$.
It is an equivariant subset of $S \times S \times A \times E$ and so it is orbit-finite.
Hence we can decide emptiness and obtain representatives if it is non-empty.

Constructing the hypothesis happens in the same way as before (Section~\ref{sec:overview}), where we note the state space is orbit-finite since it is a quotient of $S$.
Moreover the function $row$ is equivariant, so all structure ($Q_0$, $F$ and $\delta$) is equivariant as well.

The representation given above is not the only way to represent nominal sets.
For example, first-order definable sets can be used as well \cite{nlambda-paper}.
From now on we assume to have set theoretic primitives so that each line in Figure~\ref{fig:alg} is well defined.

\subsection{Correctness}
To prove correctness we only have to prove that the algorithm terminates, that is, only finitely many hypotheses will be produced. Correctness follows trivially from termination since the last step of the algorithm is an equivalence query to the teacher inquiring whether an hypothesis automaton accepts the target language.  
We start out by listing some facts about observation tables.
\begin{lemma}
	\label{lem:row_equiv}
	\label{lem:nerode_approx}
	The relation $\sim$ is an equivariant equivalence relation.
	Furthermore, for all $u, v \in S$ we have that $u \equiv_\lang v$ implies $u \sim v$.
\end{lemma}
Lemma \ref{lem:row_equiv} implies that at any stage of the algorithm the number of orbits of $S / {\sim}$ does not exceed the number of orbits of the minimal acceptor with state space $A^\star / {\equiv_\lang}$ (recall that $\equiv_\lang$ is the nominal Myhill-Nerode equivalence relation). Moreover, the following lemma shows that the dimension of the state space never exceeds the dimension of the minimal acceptor. Recall that the dimension is the maximal size of the support of any state, which is different than the number of orbits.  
\begin{lemma}
	\label{lem:nerode_support}
	We have $\supp([u]_{\sim}) \subseteq \supp([u]_{\equiv_\lang}) \subseteq \supp(u)$ for all $u \in S$.
\end{lemma}
\begin{lemma}
	The automaton constructed from a closed and consistent table is minimal.
\end{lemma}
\begin{proof}
Follows from the categorical perspective given in \cite{JacobsSilva14}.
\end{proof}
We note that the constructed automaton is consistent with the table (we use that the set $S$ is prefix-closed and $E$ is suffix-closed \cite{Angluin87}).
The following lemma shows that there are no strictly ``smaller'' automata consistent with the table.
So the automaton is not just minimal, it is minimal w.r.t. the table.

\begin{lemma}
	\label{lem:minimal_wrt_table}
	Let $H$ be the automaton associated with a closed and consistent table $(S, E)$.
	If $M'$ is an automaton consistent with $(S, E)$ (meaning that $se \in \lang(M') \iff se \in \lang(H)$ for all $s \in S \cup S \Lext A$ and $e \in E$) and $M'$ has at most as many orbits as $H$,
	then there is a surjective map $f: Q_{M'} \to Q_H$. If moreover
	\begin{itemize}
		\item $M'$s dimension is bounded by the dimension of $H$, i.e. $\supp(m) \subseteq \supp(f(m))$ for all $Q_M'$, and
		\item $M'$ has no fewer local symmetries than $H$, i.e. $\pi \cdot f(m) = f(m)$ implies $\pi \cdot m = m$ for all $m \in Q_M'$,
	\end{itemize}
	then $f$ defines an isomorphism $M' \cong H$ of nominal DFAs.
\end{lemma}
\begin{proof}
	(All maps in this proof are equivariant.)
	Define a map $row' : Q_M' \to 2^E$ by restricting the language map $Q_M' \to 2^{A^\star}$ to $E$.
	First, observe that $row'(\delta'(q_0', s)) = row(s)$ for all $s \in S \cup S \Lext A$, since $\epsilon \in E$ and $M'$ is consistent with the table.
	Second, we have $\{ row'(\delta'(q_0', s)) | s \in S \} \subseteq \{ row'(q) | q \in M' \}$.

	Let $n$ be the number of orbits of $H$.
	The former set has $n$ orbits by the first observation, the latter set has at most $n$ orbits by assumption.
	We conclude that the two sets (both being equivariant) must be equal.
	That means that for each $q \in M'$ there is a $s \in S$ such that $row'(q) = row(s)$.
	We see that $row' : M' \twoheadrightarrow \{ row'(\delta'(q_0', s)) | s \} = H$ is a surjective map.
	Since a surjective map cannot increase the dimensions of orbits and the dimensions of $M'$ are bounded, we note that the dimensions of the orbits in $H$ and $M'$ have to agree.
	Similarly, surjective maps preserve local symmetries.
	This map must hence be an isomorphism of nominal sets.
	Note that $row'(q) = row'(\delta'(q_0', s))$ implies $q = \delta'(q_0', s)$.

	It remains to prove that it respects the automaton structures.
	It preserve the initial state: $row'(q_0') = row(\delta'(q_0', \epsilon)) = row(\epsilon)$.
	Now let $q \in M'$ be a state and $s \in S$ such that $row'(q) = row(s)$.
	It preserves final states: $q \in F' \iff row'(q)(\epsilon) = 1 \iff row(s)(\epsilon) = 1$.
	Finally, it preserves the transition structure:
	\begin{align*}
		row'(\delta'(q, a)) & = row'(\delta'(\delta'(q_0', s), a)) 
		                     = row'(\delta'(q_0', sa)) \\
		                    & = row(sa) = \delta(row(s), a)
		                    \qedhere
	\end{align*}
\end{proof}
The above proof is an adaptation of Angluin's proof for automata over sets. We will now prove termination of the algorithm by proving that all steps are productive.

\begin{theorem}
	\label{thm:termination}
	The algorithm terminates and is hence correct.
\end{theorem}
\begin{proof}
	Provided that the if-statements and set operations terminate, we are left proving that
	the algorithm adds (orbits of) rows and columns only finitely often.
	We start by proving that a table can be made closed and consistent in finite time.

	If the table is not closed, we find a row $s_1 \in S \Lext A$ such that $row(s_1) \neq row(s)$ for all $s \in S$.
	The algorithm then adds the orbit containing $s_1$ to $S$.
	Since $s_1$ was nonequivalent to all rows, we find that $S \cup \orb(t) / {\sim}$ has strictly more orbits than $S / {\sim}$.
	Since orbits of $S/\sim$ cannot be more than those of $A^\star / {\equiv_\lang}$, this happens finitely often.

	Columns are added in case of an inconsistency.
	Here the algorithm finds two elements $s_1, s_2 \in S$ with $row(s_1) = row(s_2)$ but $row(s_1 a e) \neq row(s_2 a e)$ for some $a \in A$ and $e \in E$.
	Adding $a e$ to $E$ will ensure that $row'(s_1) \neq row'(s_2)$ ($row'$ is the function belonging to the updated observation table).
	If the two elements $row'(s_1), row'(s_2)$ are in different orbits, the number of orbits is increased.
	If they are not in the same orbit, we have $row'(s_2) = \pi \cdot row'(s_1)$ for some permutation $\pi$. Using $row(s_1) = row(s_2)$ and $row'(s_1) \neq row'(s_2)$ we have:
	\begin{align*}
	row(s_1)   =    \pi \cdot row(s_1) \qquad
	row'(s_1)  \neq \pi \cdot row'(s_1)
	\end{align*}
	Consider all such $\pi$ and
	suppose there is a $\pi$ and $x \in \supp(row(s_1))$ such that $\pi \cdot x \notin \supp(row(s_1))$.
	Then we find that $\pi \cdot x \in \supp(row'(s_1))$, and so the support of the row has grown.
	By Lemma~\ref{lem:nerode_support} this happens finitely often.
	Suppose such $\pi$ and $x$ do not exist, then we consider the finite group $R = \{\rho|_{\supp([s_1]_{\sim})} \,|\, row(s_1) = \rho \cdot row(s_1) \}$.
	We see that $\{\rho|_{\supp([s_1]_{\sim})} \,|\, row'(s_1) = \rho \cdot row'(s_1) \}$ is a proper subgroup of $R$.
	So, adding a column in this case decreases the size of the group $R$, which can happen only finitely often. In this case a local symmetry is removed.

	In short, the algorithm will succeed in producing a hypothesis in each round.
	It remains to prove that it needs only finitely many equivalence queries.

	Let $(S, E)$ be the closed and consistent table and $H$ its corresponding hypothesis.
	If it is incorrect a second hypothesis $H'$ will be constructed which is consistent with the old table $(S, E)$.
	The two hypotheses are nonequivalent, as $H'$ will handle the counter example correctly and $H$ does not.
	Therefore, $H'$ will have at least one orbit more, one local symmetry less, or one orbit will have strictly bigger dimension (Lemma~\ref{lem:minimal_wrt_table}),
	all of which can only happen finitely often.
\end{proof}

\newcommand{\NomCell}[1]{
\begin{cases}
	1 & {#1} \\[-0.7ex]
    0 & \text{else}
\end{cases}
}

\begin{figure*}[t]
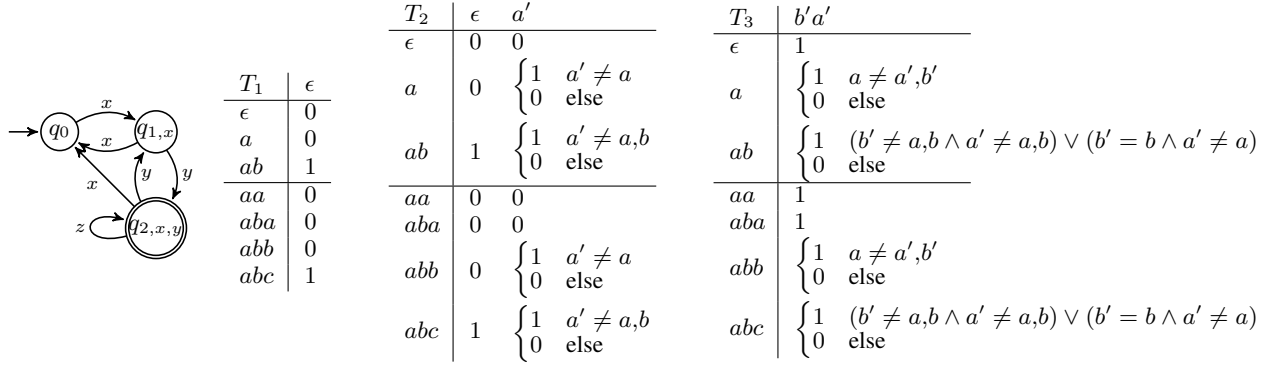

\begin{center}
	\begin{tabular}{m{.15\linewidth}m{.1\linewidth}m{.22\linewidth}m{.62\linewidth}}
		\raggedright
		\begin{automaton}
		\node[initial,state] (q0) {$q_0$};
		\node[state,right of=q0] (q1) {$q_{1,x}$};
		\node[state,accepting,below of=q1] (q2) {$q_{2, x, y}$};
		\path
		(q0) edge[bend left] node[trlab,above] {$x$} (q1)
		(q1) edge[bend left] node[trlab,above] {$x$} (q0)
		(q1) edge[bend left] node[trlab,right] {$y$} (q2)
		(q2) edge[bend left] node[trlab,right] {$y$} (q1)
		(q2) edge node[trlab,below left] {$x$} (q0)
		(q2) edge[loop left] node[trlab,left] {$z$} (q2);
		\end{automaton}
		&
		\raggedright
		\begin{tabular}{R{3ex}|C{1ex}}
			T_1 & \epsilon   \\
			\hline
			\epsilon & 0 \\
			a        & 0 \\
			ab       & 1 \\
			\hline
			aa       & 0 \\
			aba      & 0 \\
			abb      & 0 \\
			abc      & 1
		\end{tabular}
	&
		\raggedright
		\begin{tabular}{R{3ex}|L{1ex}R{12ex}}
			T_2 & \epsilon   & a'  \\
			\hline
			\epsilon & 0 & 0   \\
			a        & 0 & \NomCell{a' \neq a} \\[2ex]
			ab       & 1 & \NomCell{a' \neq a{,} b} \\[2ex]
			\hline
			aa       & 0 & 0   \\
			aba      & 0 & 0   \\
			abb      & 0 & \NomCell{a' \neq a} \\[2ex]
			abc      & 1 & \NomCell{a' \neq a{,} b}
		\end{tabular}
		&
		\raggedright
		\begin{tabular}{R{3ex}|R{15ex}}
			T_3  & b'a'  \\
			\hline
			\epsilon  & 1    \\
			a         & \NomCell{a \neq a'{,}b'} \\[2ex]
			ab        & \NomCell{(b' \neq a{,}b \wedge a' \neq a{,}b) \vee (b' = b \wedge a' \neq a)} \\
			\hline
			aa        & 1 \\
			aba       & 1 \\
			abb       & \NomCell{a \neq a'{,}b'} \\[2ex]
			abc       & \NomCell{(b' \neq a{,}b \wedge a' \neq a{,}b) \vee (b' = b \wedge a' \neq a)}
		\end{tabular}
	\end{tabular}
\end{center}
\caption{Example automaton to be learnt and three subsequent tables computed by \nLStar. In the automaton, $x,y,z$ denote distinct atoms. In $T_3$ we only show a relevant column.}
\label{fig:ex}
\end{figure*}

We remark that all the lemmas and proofs as above are close to the original ones of Angluin. 
However, two things are crucially different.
First, adding a column does not always increase the number of (orbits of) states.
It can happen that by adding a column a bigger support is found or that a local symmetry is broken.
Second, the new hypothesis does not necessarily have more states, again it might have bigger dimensions or less local symmetries.

From the proof Theorem~\ref{thm:termination} we observe moreover that the way we handle counterexamples is not crucial.
Any other method which ensures a non-equivalent hypothesis will work.
In particular our algorithm is easily adapted to include optimizations such as the ones in \cite{RivestSchapire93} and \cite{MalerPnueli95}, where counterexamples are added as columns.

\subsection{Example}
Consider the target automaton in Figure~\ref{fig:ex} and an observation table $T_1$ at some stage during the algorithm.
We remind the reader that the table is represented in a symbolic way: the sequences in the rows and columns stand for whole orbits and the cells denote functions from the product of the orbits to $2$.
Since the cells can consist of multiple orbits, where each orbit is allowed to have a different value, we use a formula to specify which orbits have a $1$.

The table $T_1$ at some stage of the algorithm has to be checked for closedness and consistency.
We note that it is definitely closed.
For consistency we check the rows $row(\epsilon)$ and $row(a)$ which are equal.
Observe, however, that $row(\epsilon b)(\epsilon)=0$ and $row(ab)(\epsilon)=1$, so we have an inconsistency.
The algorithm adds the orbit $\orb(b)$ as column and extends the table, obtaining $T_2$.
We note that, in this process, the number of orbits did grow, as the two rows are split.
Furthermore we see that both $row(a)$ and $row(ab)$ have empty support in $T_1$, but not in $T_2$,
because $row(a)(a')$ depends on $a'$ being equal or different from $a$, similarly for $row(ab)(a')$.

The table $T_2$ is still not consistent as we see that $row(ab) = row(ba)$ but $row(abb)(c) = 1$ and $row(bab)(c) = 0$.
Hence the algorithm adds the columns $\orb(bc)$, obtaining table $T_3$.
We note that in this case, no new orbits are obtained and no support has grown.
In fact, the only change here is that the local symmetry between $row(ab)$ and $row(ba)$ is removed.
This last table, $T_3$, is closed and consistent and will produce the correct hypothesis.

\subsection{Query Complexity}
\label{ssec:lstar-compl}
In this section, we will analyse the number of queries made by the algorithm in the worst case.
Let $M$ be the minimal target automaton with $n$ orbits and of dimension $k$.
We will use $\log$ in base two.

\begin{lemma}
	\label{lem:bound_eq}
	The number of equivalence queries $E_{n,k}$ is $O(n k \log k)$.
\end{lemma}
\begin{proof}
	By Lemma~\ref{lem:minimal_wrt_table} each hypothesis will be either
	1) bigger in the number of orbits, which is bounded by $n$, or
	2) bigger in the dimension of an orbit, which is bounded by $k$ or
	3) smaller in local symmetries of an orbit.
	For the last part we want to know how long a subgroup series of the permutation group $S_k$ can be.
	This is bounded by the number of divisors of $k!$, as each subgroup divides the order of the group.
	We can easily bound the number of divisors of any $m$ by $\log m$ and so
	one can at take a subgroup at most $k \log k$ times when starting with $S_k$.

	Since the hypothesis will grow monotonically in the number of orbits and for each orbit will grow monotonically w.r.t.\ the remaining two dimensions, the number of equivalence queries is bound by $n + n (k + k \log k)$.
\end{proof}

Next we will give a bound for the size of the table.

\begin{lemma}
	The table has at most $n + m E_{n,k}$ orbits in $S$ with sequences of at most length $n + m$,
	where $m$ is the length of the longest counter example given by the teacher.
	The table has at most $n (k + k \log k + 1)$ orbits in $E$ of at most length $n (k + k \log k + 1)$
\end{lemma}
\begin{proof}
	In the termination proof we noted that rows are added at most $n$ times.
	In addition (all prefixes of) counter examples are added as rows which add another $m E_{n,k}$ rows.
	Obviously counter examples are of length at most $m$ and are extended at most $n$ times, making the length at most $m + n$ in the worst case.

	For columns we note that one of three dimensions approaches a bound similarly to the proof of Lemma~\ref{lem:bound_eq}.
	So at most $n (k + k \log k + 1)$ columns are added.
	Since they are suffix closed, the length is at most $n (k + k \log k + 1)$.
\end{proof}

Let $p$ and $l$ denote respectively the dimension and the number of orbits of $A$.

\begin{lemma}
The number of orbits in the lower part of the table, $S \Lext A$, is bounded by $(n + m E_{n,k}) l f_\EAlph( p(n+m), p)$.
\end{lemma}
\begin{proof}
Any sequence in $S$ is of length at most $n+m$, so it contains at most $p(n+m)$ distinct atoms. When we consider $S \Lext A$, the extension can either reuse atoms from those $p(n+m)$, or none at all. Since the extra letter has at most $p$ distinct atoms, the set $\EAlph^{(p(n+m)) } \times \EAlph^{(p)}$ gives a bound $f_\EAlph(p(n+m),p)$ for the number of orbits of $O_S \times O_A$, with $O_X$ an orbit of $X$. Multiplying by the number of such ordered pairs, namely $(n + m E_{n,k})l$, gives a bound for $S \Lext A$.
\end{proof}

Let $C_{n,k,m} = (n + m E_{n,k}) (l f_\EAlph(p(n+m), p) + 1) n (k + k \log k + 1)$ be the maximal number of cells in the table. We note that this number is polynomial in $k, l, m$ and $n$ but not in $p$.

\begin{corollary}
	The number of membership queries is bounded by $C_{n,k,m} f_\EAlph(p (n + m), p n (k + k \log k + 1))$.
\end{corollary}
\section{Learning non-deterministic nominal automata}
\label{sec:nondet}
In this section, we introduce a variant of \nLStar, which we call \nNLStar, where the learnt automaton is non-deterministic. It will be based on \NLStar\cite{BolligHKL09}, an Angluin-style algorithm for learning NFAs. The algorithm is shown in Figure~\ref{fig:nfa-alg}. We first illustrate \NLStar, then we discuss its extension to nominal automata.

\NLStar\ crucially relies on the use of \emph{residual finite-state automata} (RFSA) \cite{DenisLT02}, which are NFAs admitting \emph{unique minimal canonical representatives}. The states of this automaton correspond to Myhill-Nerode right-congruence classes, but can be exponentially smaller than the corresponding minimal DFA: \emph{composed} states, language-equivalent to sets of other states, can be dropped.
\begin{figure}[t]
\begin{codebox}
\Procname{$\proc{\NLStar\ learner}$}
\li $S,E \gets \{\epsilon\}$
\li \Repeat
\li \While $(S, E)$ is not RFSA-closed or not RFSA-consistent\label{line:nfa-checks}
\li \If $(S, E)$ is not RFSA-closed
\li \Then\label{line:begin-nfa-closed}
find $s\in S,a \in A$ such that
\zi \qquad $row(sa) \in PR(S,E) \setminus PR^{\top}(S,E)$
\label{line:nfa-clos-witness}
\li $S\gets S\cup \{sa\}$ \label{line:nfa-addrow-clos}
\End\label{line:nfa-end-closed}
\li \If $(S, E)$ is not RFSA-consistent\label{line:nfa-begin-const}
\li \Then find $s_1,s_2 \in S$, $a \in A$, and $e\in E$ such that
\zi \qquad $row(s_1) \rowincl row(s_2)$ and
\zi \qquad $\lang(s_1 a e) = 1$, $\lang(s_2 a e) = 0$
\label{line:nfa-cons-witness}
\li $E\gets E\cup \{ae\}$ \label{line:nfa-addcol-cons}
\End
\label{line:nfa-end-const}
\li Make the conjecture $N(S,E)$
\label{line:nfa-conj}
\li \If the Teacher replies  {\bf no}, with a counter-example $t$\label{line:nfa-counter-ex}
\li \Then $E\gets E \cup \suff(t)$ \label{line:nfa-addcol-ex}
\End
\li \Until the Teacher replies {\bf yes} to the conjecture $N(S,E)$.
\li \Return $N(S,E)$
\end{codebox}
\caption{Bollig et al.'s algorithm for learning NFAs~\cite{BolligHKL09}}\label{fig:nfa-alg}
\end{figure}
The algorithm \NLStar\ equips the observation table $(S,E)$ with a union operation, allowing for the detection of \emph{composed} and \emph{prime} rows.
\begin{definition}
Let $(row(s_1) \rowunion row(s_2))(e) = row(s_1)(e) \lor row(s_1)(e)$ (regarding cells as booleans). This operation induces an ordering between rows: $row(s_1) \rowincl row(s_2)$ whenever $row(s_1)(e) = 1$ implies $row(s_2)(e) = 1$, for all $e \in E$.
\label{def:rows-jsl}
\end{definition}
A row $row(s)$ is composed if $row(s) = row(s_1) \rowunion \dots \rowunion row(s_n)$, for $row(s_i) \neq row(s)$. Otherwise it is prime. We denote by $PR^\top(S,E)$ the rows in the top part of the table (ranging over $S$) which are prime w.r.t.\ the whole table (not only w.r.t.\ the top part). We write $PR(S,E)$ for all the prime rows of $(S,E)$.

As in \LStar, states of hypothesis automata will be rows of $(S,E)$ but, as the aim is to construct a minimal RFSA, only prime rows are picked. New notions of closedness and consistency are introduced, to reflect features of RFSAs.
\begin{definition}
A table $(S,E)$ is:
\begin{itemize}
	\item \emph{RFSA-closed} if, for all $t \in S \Lext A$, $row(t) = \Rowunion \{ row(s) \in PR^\top(S,E) \mid row(s) \rowincl row(t)\}$;	\item \emph{RFSA-consistent} if, for all $s_1,s_2 \in S$ and $a \in A$, $row(s_1) \rowincl row(s_2)$ implies $row(s_1a) \rowincl row(s_2a)$.
\end{itemize}
\end{definition}
If $(S,E)$ is not RFSA-closed, then there is a row in the bottom part of the table which is prime, but not contained in the top part. This row is then added to $S$ (line \ref{line:nfa-clos-witness}). If $(S,E)$ is not RFSA-consistent, then there is a suffix which does not preserve the containment of two existing rows, so those rows are actually incomparable. A new column is added to distinguish those rows (line \ref{line:nfa-cons-witness}). Notice that counterexamples supplied by the teacher are added \emph{to columns} (line \ref{line:nfa-addcol-ex}). Indeed, in \cite{BolligHKL09} it is shown that treating the counterexamples as in the original \LStar, namely adding them to rows, does not lead to a terminating algorithm.

\begin{definition}
Given a RFSA-closed and RFSA-consistent table $(S,E)$, the conjecture automaton is $N(S,E) = (Q,Q_0,F,\delta)$, where:
\begin{itemize}
	\item $Q = PR^\top(S,E)$;
	\item $Q_0 = \{ r \in Q \mid r \rowincl row(\epsilon)\}$;
	\item $F = \{r \in Q \mid r(\epsilon) = 1\}$;
	\item the transition relation $\delta \subseteq Q \times A \times Q$ is given by $\delta(row(s),a) = \{ r \in Q \mid r \rowincl row(sa) \}$.
\end{itemize}
\label{def:nfa-conj}
\end{definition}
As observed in \cite{BolligHKL09}, $N(S,E)$ is not necessarily a RFSA, but it is a canonical RFSA if it is consistent with $(S,E)$. If the algorithm terminates, then $N(S,E)$ must be consistent with $(S,E)$, which ensures correctness. The termination argument is more involved than \LStar, but still it relies on the minimal DFA.

Developing an algorithm to learn nominal NFAs is not an obvious extension of \NLStar: non-deterministic nominal languages strictly contain nominal regular languages, so it is not clear what the developed algorithm should be able to learn. To deal with this, we introduce a nominal notion of RFSAs. They are a \emph{proper subclass} of nominal NFAs, because they recognize nominal regular languages. Nonetheless, they are more succinct than nominal DFAs. 

\subsection{Nominal residual finite-state automata}
\label{subsec:nrfsa}
Let $\lang$ be a nominal regular language and let $u$ be a finite string. The \emph{derivative} of $\lang$ w.r.t.\ $u$ is
$
	\lder{u}\lang = \{ v \in A^\star \mid uv \in \lang \}.
$
A set $\lang' \subseteq \EAlph^*$ is a \emph{residual} of $\lang$ if there is $u$ with $\lang' = \lder{u}\lang$.
Note that a residual might not be equivariant, but it does have a finite support. We write $R(\lang)$ for the set of residuals of $\lang$. Residuals form an orbit-finite nominal set: they are in bijection with the state-space of the minimal nominal DFA for $\lang$.

A \emph{nominal residual finite-state automaton} for $\lang$ is a nominal NFA whose states are subsets of such minimal automaton. Given a state $q$ of an automaton, we write $\lang(q)$ for the set of words leading from $q$ to a set of states containing a final one.
\begin{definition}
A \emph{nominal residual finite-state automaton} (nominal RFSA) is a nominal NFA $\autom$ such that $\lang(q) \in R(\lang(\autom))$, for all $q \in Q_\autom$.
\end{definition}
Intuitively, all states of a nominal RSFA recognize residuals, but not all residuals are recognized by a single state: there may be a residual $\lang'$ and a set of states $Q'$ such that $\lang' = \bigcup_{q \in Q'} \lang(q)$, but no state $q'$ is such that $\lang(q') = \lang'$. A residual $\lang'$ is called \emph{composed} if it is equal to the union of the components it strictly contains, explicitly
\[
\lang' = \Union \{ \lang'' \in R(\lang) \mid \lang'' \subsetneqq \lang' \} \enspace ;
\]
otherwise it is called \emph{prime}.
In an ordinary RSFA, composed residuals have finitely-many components. This is not the case in a nominal RFSA. However, the set of components of $\lang'$ always has a finite support, namely $\supp(\lang')$.

The set of prime residuals $PR(\lang)$ is an orbit-finite nominal set, and can be used to define a \emph{canonical} nominal RFSA for $\lang$, which has the minimal number of states and the maximal number of transitions. This can be regarded as obtained from the minimal nominal DFA, by removing composed states and adding all initial states and transitions that do not change the recognized language. This automaton is necessarily unique.
\begin{lemma}
\label{lem:can-rfsa}
Let the \emph{canonical} nominal RSFA of $\lang$ be $(Q,Q_0,F,\delta)$ such that:
\begin{itemize}
	\item $Q = PR(\lang)$;
		\item $Q_0 = \{ \lang ' \in Q \mid \lang' \subseteq \lang\}$;
	\item $F = \{ \lang' \in Q \mid \epsilon \in \lang' \}$;
	\item $\delta(\lang_1,a) = \{ \lang_2 \in Q \mid \lang_2 \subseteq \lder{a}\lang_1 \}$.
\end{itemize}
It is a well-defined nominal NFA accepting $\lang$.
\end{lemma}

\subsection{\nNLStar}

Our nominal version of \NLStar\ again makes use of an observation table $(S,E)$ where $S$ and $E$ are equivariant subsets of $A^*$ and $row$ is an equivariant function. As in the basic algorithm, we equip $(S,E)$ with a union operation $\rowunion$ and row containment relation $\rowincl$, defined as in Definition~\ref{def:rows-jsl}. It is immediate to verify that $\rowunion$ and $\rowincl$ are equivariant.

Our algorithm is a simple modification of the algorithm in Figure~\ref{fig:nfa-alg}, where a few lines are replaced:
\[
\begin{array}{ll}
	6' &S \gets S \cup \orb(sa) \\
	9' &E \gets E \cup \orb(ae) \\
	12' &E \gets E \cup \suff(\orb(t))
\end{array}
\]
Switching to nominal sets, several decidability issues arise. The most critical one is that rows may be the union of infinitely many component rows, as happens for residuals of nominal languages, so finding all such components can be challenging. We adapt the notion of composed to rows: $row(t)$ is composed whenever
\[
	row(t) = \Rowunion \{ row(s) \mid row(s) \rowincls row(t) \} \enspace .
\]
where $\rowincls$ is \emph{strict} row inclusion; otherwise $row(t)$ is prime.

We now check that relevant parts of our algorithm terminate.
\paragraph{Row containment check.}
The basic containment check $row(s) \rowincl row(t)$ is decidable, as $row(s)$ and $row(t)$ are supported by the finite supports of $s$ and $t$ respectively.
\paragraph{line \ref{line:nfa-checks}: RFSA-closedness and RFSA-consistency checks.}

We first show that prime rows form orbit-finite nominal sets.
\begin{lemma}
$PR(S,E)$, $PR^\top(S,E)$ and $PR(S,E) \setminus PR^\top(S,E)$ are orbit-finite nominal sets.
\label{lem:pr-ns}
\end{lemma}
Consider now RFSA-closedness. It requires computing the set $C(row(t))$ of components of $row(t)$ contained in $PR^{\top}(S,E)$ (possibly including $row(t)$). This may not be equivariant under permutations $Perm(\EAlph)$, but it is if we pick a subgroup.
\begin{lemma}
The set $C(row(t))$ has the following properties:
\begin{enumerate}
	\item $\supp(C(row(t))) \subseteq \supp(row(t))$. \label{supp}
	\item it is equivariant and orbit-finite under the action of the group
	\[
		G_t = \{ \pi \in Perm(\EAlph) \mid \restr{\pi}{\supp(row(t))} = id \}
	\]
	of permutations fixing $\supp(row(t))$. \label{orb-fin}
\end{enumerate}
\label{lem:comp-prop}
\end{lemma}
We established that $C(row(t))$ can be effectively computed, and the same holds for $\Rowunion C(row(t))$. In fact, $\Rowunion$ is equivariant w.r.t\ the whole $Perm(\EAlph)$ and then, in particular, w.r.t.\ $G_t$, so it preserves orbit-finiteness. Now, to check $row(t) = \Rowunion C(row(t))$, we can just pick one representative of every orbit of $S \Lext A$, because we have $C(\pi \cdot row(t)) = \pi \cdot C(row(t))$ and permutations distribute over $\rowunion$, so permuting both sides of the equation gives again a valid equation.

For RFSA-consistency, consider the two sets:
\begin{align*}
	N &= \{ (s_1,s_2) \in S \times S \mid row(s_1) \rowincl row(s_2) \} \\
	M &= \{ (s_1,s_2) \in S \times S \mid \forall a \in A : row(s_1a) \rowincl row(s_2a) \}
\end{align*}
They are both orbit-finite nominal sets, by equivariance of $row$, $\rowincl$ and $A$. We can check RFSA-consistency in finite time by picking orbit representatives from $N$ and $M$. For each representative $n \in N$, we look for a representative $m \in M$ and a permutation $\pi$ such that $n = \pi \cdot m$. If no such $m$ and $\pi$ exist, then $n$ does not belong to any orbit of $M$, so it violates RFSA-consistency.
\paragraph{lines \ref{line:nfa-clos-witness} and \ref{line:nfa-cons-witness}: finding witnesses for violations.}
We can find witnesses by comparing orbit representatives of orbit-finite sets, as we did with RFSA-consistency. Specifically, we can pick representatives in $S \times A$ and $S \times S \times A \times E$ and check them against the following orbit-finite nominal sets:
\begin{itemize}
	\item $\{ (s,a) \in S \times A \mid row(sa) \in PR(S,E) \setminus PR^{\top}(S,E) \}$;
	\item $\{ (s_1,s_2,a,e) \in S \times S \times A \times E \mid row(s_1a)(e) = 1, row(s_2a)(e) = 0, row(s_1) \rowincl row(s_2) \}$;
\end{itemize}

\subsection{Correctness}
Now we prove correctness and termination of the algorithm. First, we prove that hypothesis automata are nominal NFAs.
\begin{lemma}
\label{lem:nlstarwelldef}
The hypothesis automaton $N(S,E)$ (see Definition~\ref{def:nfa-conj}) is a nominal NFA.
\end{lemma}
$N(S,E)$, as in ordinary \NLStar, is not always a nominal RFSA. However, we have the following.

\begin{theorem}
If the table $(S,E)$ is RFSA-closed, RFSA-consistent and $N(S,E)$ is consistent with $(S,E)$, then $N(S,E)$ is a canonical nominal RFSA.
\label{th:nlstar-rfsa}
\end{theorem}
This is proved in \cite{NLTR} for ordinary RFSAs, using the standard theory of regular languages. The nominal proof is exactly the same, using derivatives of nominal regular languages and nominal RFSAs as defined in Section~\ref{subsec:nrfsa}.
\begin{lemma}
The table $(S,E)$ cannot have more than $n$ orbits of distinct rows, where $n$ is the number of orbits of the minimal nominal DFA for the target language.
\label{lem:rows-min-dfa}
\end{lemma}
\begin{proof}
Rows are residuals of $\lang$, which are states of the minimal nominal DFA for $\lang$, so orbits cannot be more than $n$.
\end{proof}
\begin{theorem}
The algorithm \nNLStar\ terminates and returns the canonical nominal RFSA for $\lang$.
\label{th:nom-nl-term}
\end{theorem}
\begin{proof}
If the algorithm terminates, then it must return the canonical nominal RFSA for $\lang$ by Theorem~\ref{th:nlstar-rfsa}. We prove that a table can be made RFSA-closed and RFSA-consistent in finite time. This is similar to the proof of Theorem~2 and is inspired by the proof \cite[Theorem 3]{NLTR}.

	If the table is not RFSA-closed, we find a row $s \in S \Lext A$ such that $ row(s) \in PR(S,E) \setminus PR^{\top}(S,E)$. The algorithm then adds $\orb(s)$ to $S$. Since $s$ was nonequivalent to all upper prime rows, and thus from all the rows indexed by $S$, we find that $S \cup \orb(t) / {\sim}$ has strictly more orbits than $S / {\sim}$ (recall that $s \sim t \iff row(s) = row(t)$). This addition can only be done finitely-many times, because the number of orbits of $S / {\sim}$ is bounded, by Lemma~\ref{lem:rows-min-dfa}.

Now, the case of RFSA-consistency needs some additional notions.
Let $R$ be the (orbit-finite) nominal set of all rows, and let
$
	I = \{ (r,r') \in R \times R \mid r \rowincls r'\}
$
be the set of all inclusion relations among rows. The set $I$ is orbit-finite. In fact, consider
\[
	J = \{ (s,t) \in ( S \cup S\Lext A) \times ( S \cup S\Lext A) \mid row(s) \rowincls row(t) \}
\]
This set is an equivariant, thus orbit-finite, subset of $( S \cup S\Lext A) \times ( S \cup S\Lext A)$. The set $I$ is the image of $J$ via $row \times row$, which is equivariant, so it preserves orbit-finiteness.

Now, suppose the algorithm finds two elements $s_1, s_2 \in S$ with $row(s_1) \rowincl row(s_2)$ but $row(s_1 a)(e) = 1$ and $row(s_2 a)(e) = 0$ for some $a \in A$ and $e \in E$. Adding a column to fix RFSA-consistency may: {\bf C1)} increase orbits of $(S \cup S \Lext A) / {\sim}$, or; {\bf C2)} decrease orbits of $I$, or; {\bf C3)} decrease local symmetries/increase dimension of one orbit of rows. In fact, if no new rows are added ({\bf C1}), we have two cases.
\begin{itemize}
	\item 
	If $row(s_1) \rowincls row(s_2)$, i.e., $(row(s_1),row(s_2)) \in I$, then $row'(s_1) \not \rowincls row'(s_2)$, where $row'$ is the new table. Therefore the orbit of $(row'(s_1),row'(s_2))$ is not in $I$. Moreover, $row'(s) \rowincls row'(t)$ implies $row(s) \rowincls row(t)$ (as no new rows are added), so no new pairs are added to $I$. Overall, $I$ has less orbits ({\bf C2}).
	\item 
	If $row(s_1) = row(s_2)$, then we must have $row(s_1) = \pi \cdot row(s_1)$, for some $\pi$, because line~\ref{line:begin-nfa-closed} forbids equal rows in different orbits. In this case $row'(s_1) \neq \pi \cdot row'(s_1)$ and we can use part of the proof of Theorem~\ref{thm:termination} to see that the orbit of $row'(s_1)$ has bigger dimension or less local symmetries than that of $row(s_1)$ ({\bf C3}).
\end{itemize}
Orbits of $(S \cup S \Lext A)/{\sim}$ and of $I$ are finitely-many, by Lemma~\ref{lem:rows-min-dfa} and what we proved above. Moreover, local symmetries can decrease finitely-many times, and the dimension of each orbit of rows is bounded by the dimension of the minimal DFA state-space. Therefore all the above changes can happen finitely-many times.

We have proved that the table eventually becomes RFSA-closed and RFSA-consistent. Now we prove that a finite number of equivalence queries is needed to reach the final hypothesis automaton. To do this, we cannot use a suitable version of Lemma~\ref{lem:minimal_wrt_table}, because this relies on $N(S,E)$ being consistent with $(S,E)$, which in general is not true (see \cite{NLTR} for an example of this). We can, however, use an argument similar to that for RFSA-consistency, because the algorithm adds columns in response to counterexamples. Let $w$ the counterexample provided by the teacher. When $12'$ is executed, the table must change. In fact, by \cite[Lemma~2]{NLTR}, if it does not, then $w$ is already correctly classified by $N(S,E)$, which is absurd. We have the following cases: 
\begin{description}
	\item[E1:] orbits of $(S \cup S \Lext A)/{\sim}$ increase ({\bf C1}), or;
	\item[E2:] either: orbits in $PR(S,E)$ increase, or any of the following happens: orbits in $I$ decrease ({\bf C2}), local symmetries/ dimension of an orbit of rows change ({\bf C3}).
\end{description}
In fact, if {\bf E1} does not happen and $PR(S,E)$, $I$ and local symmetries/dimension of orbits of rows do not change, the automaton $\autom$ for the new table coincides with $N(S,E)$. But $N(S,E) = \autom$ is a contradiction, because $\autom$ correctly classifies $w$ (by \cite[Lemma~2]{NLTR}, as $w$ now belongs to columns), whereas $N(S,E)$ does not. Both {\bf E1} and {\bf E2} can only happen finitely-many times.
\end{proof}

\subsection{Query complexity}
We now give bounds for the number of equivalence and membership queries needed by \nNLStar. Let $n$ be the number of orbits of the minimal DFA $M$ for the target language and let $k$ be the dimension (i.e., the size of the maximum support) of its nominal set of states.
\begin{lemma}
The number of equivalence queries $E'_{n,k}$ is $O(n^2 f_\EAlph(k,k) + nk \log k)$.
\label{lem:nfa-eq-queries}
\end{lemma}
\begin{proof}
In the proof of Theorem~\ref{th:nom-nl-term}, we saw that equivalence queries lead to more orbits in $(S \cup S\Lext A)/{\sim}$, in $PR(S,E)$, less orbits in $I$ or less local symmetries/bigger dimension for an orbit. Clearly the first two ones can happen at most $n$ times. We now estimate how many times $I$ can decrease. Suppose $(S \cup S\Lext A)/{\sim}$ has $d$ orbits and $h$ orbits are added to it. Recall that, given an orbit $O$ of rows of dimension at most $m$, $f_\EAlph(m,m)$ is an upper bound for the number of orbits in the product $O \times O$. Since  the support of rows is bounded by $k$, we can give a bound for the number of orbits added to $I$: $dhf_\EAlph(k,k)$, for new pairs $r \rowincls r'$ with $r$ in a new orbit of rows and $r'$ in an old one (or viceversa); plus $(h(h-1)/2)f_\EAlph(k,k)$, for $r$ and $r'$ both in (distinct) new orbits; plus $hf_\EAlph(k,k)$, for $r$ and $r'$ in the same new orbit. 
Notice that, if $PR(S,E)$ grows but $(S \cup S\Lext A)/{\sim}$ does not, $I$ does not increase. By Lemma~\ref{lem:rows-min-dfa}, $h,d \leq n$, so $I$ cannot decrease more than $(n^2 + n(n-1)/2 + n)f_\EAlph(k,k)$ times. 

Local symmetries of an orbit of rows can decrease at most $k \log k$ times (see proof of Lemma~\ref{lem:bound_eq}), and its dimension can increase at most $k$ times. Therefore $n(k + \log k)$ is a bound for all the orbits of rows, which are at most $n$, by Lemma~\ref{lem:rows-min-dfa}. Summing up, we get the main result.  
\end{proof}
\begin{lemma}
Let $m$ be the length of the longest counterexample given by the teacher. Then the table has:
\begin{itemize}
	\item at most $n$ orbits in $S$, with words of length at most $n$;
	\item at most $mE'_{n,k}$ orbits in $E$, with words of length at most $mE'_{n,k}$.
\end{itemize}
\end{lemma}
\begin{proof}
By Lemma~\ref{lem:rows-min-dfa}, the number of orbits of rows indexed by $S$ is at most $n$. Now, notice that line~\ref{line:nfa-clos-witness} does not add $\orb(sa)$ to $S$ if $sa \in S$, and lines~\ref{line:nfa-addcol-ex}~and~\ref{line:nfa-addcol-cons} cannot identify rows, so $S$ has at most $n$ orbits. The length of the longest word in $S$ must be at most $n$, as $S = \{\epsilon\}$ when the algorithm starts, and line~$6'$ adds words with one additional symbol than those in $S$.

For columns, we note that both fixing RFSA-consistency and adding counterexamples increase the number of columns, but this can happen at most $E'_{n,k}$ times (see proof of Lemma~\ref{lem:nfa-eq-queries}). Each time at most $m$ suffixes are added to $E$.
\end{proof}
We compute the maximum number of cells as in Section~\ref{ssec:lstar-compl}. 
\begin{lemma}
The number of orbits in the lower part of the table, $S \Lext A$, is bounded by $n l f_\EAlph( pn, p)$.
\end{lemma}
Then $C_{n,k,m}' = n (l f_\EAlph(p n, p) + 1) m E'_{n,k}$ is the maximal number of cells in the table.
This bound is polynomial in $n, m$ and $l$, but not in $k$ and $p$.
\begin{corollary}
	The number of membership queries is bounded by $C_{n,k,m}' f_\EAlph(p n, p m E'_{n,k})$.
\end{corollary}
\section{Implementation and preliminary experiments}
\label{sec:impl}

Our algorithms for learning nominal automata operate on infinite sets of rows and columns, and hence it is  not immediately clear how to actually implement them on a computer. We have used NLambda~\cite{nlambda-paper}, our recently developed Haskell library designed to allow direct manipulation of infinite (but orbit-finite) nominal sets, within the functional programming paradigm. The semantics of NLambda is based on~\cite{BBKL12}, and the library itself is inspired by Fresh O'Caml~\cite{freshml}, a language for functional programming over nominal data structures with binding.

\subsection{NLambda}

NLambda extends Haskell with a new type {\tt Atoms}. Values of this type are atomic values that can be compared for equality and have no other discernible structure. They correspond to the elements of the infinite alphabet $\EAlph$ described in Section~\ref{sec:prelim}.

Furthermore, NLambda provides a unary type constructor {\tt Set}. This appears similar to the the {\tt Data.Set} type constructor from the standard Haskell library, but its semantics is markedly different: whereas the latter is used to construct finite sets, the former has {\em orbit-finite} sets as values. The new constructor {\tt Set} can be applied to a range of equality types that include {\tt Atoms}, but also the tuple type ${\tt (Atoms,Atoms)}$, the list type ${\tt [Atoms]}$, the set type ${\tt Set\ Atoms}$, and other types that provide basic infrastructure necessary to speak of supports and orbits. All these are instances of a type class {\tt NominalType} specified in NLambda for this purpose.

NLambda, in addition to all the standard machinery of Haskell, offers primitives to manipulate values of any nominal types $\tau,\sigma$:
\begin{itemize}
\item {\tt empty:Set\,$\tau$}, returns the empty set of any type;
\item {\tt atoms:Set\,Atoms}, returns the (infinite but single-orbit) set of all atoms;
\item ${\tt insert:\tau\to Set\ \tau\to Set\ \tau}$, adds an element to a set;
\item ${\tt map:(\tau\to\sigma)\to(Set\ \tau\to Set\ \sigma)}$, applies a function to every element of a set;
\item ${\tt sum:Set\ Set\ \tau\to Set\ \tau}$, computes the union of a family of sets;
\item ${\tt isEmpty:Set\ \tau\to Formula}$, checks whether a set is empty.
\end{itemize}
The type {\tt Formula} takes the role of a Boolean type. For technical reasons it is distinct from the standard Haskell type {\tt Bool}, but it provides standard logical operations such as 
\begin{align*}
\mathtt{ not}&:\mathtt{Formula\to Formula} \\
\mathtt{or}&:\mathtt{Formula\to Formula\to Formula},	
\end{align*}
as well as a conditional operator ${\tt ite:Formula\to \tau\to\tau\to\tau}$ that mimics the standard {\tt if} construction. It is also the result type of a built-in equality test on atoms, ${\tt eq:Atoms\to Atoms\to Formula}$.

Using these primitives, one builds more functions to operate on orbit-finite sets, such as a function to build singleton sets:
\begin{align*}
    &{\tt singleton}:\tau\to {\tt Set}\ \tau \\
    &{\tt singleton\ x}={\tt insert\ x\ empty}
\end{align*}
or a filtering function to select elements that satisfy a given predicate:
\begin{align*}
   &{\tt filter}:(\tau\to {\tt Formula})\to {\tt Set}\ \tau\to{\tt Set}\ \tau \\
   &{\tt filter\ p\ s} = {\tt sum}\ ({\tt map} (\lambda{\tt x}.{\tt ite}\ ({\tt p\ x})\ ({\tt singleton\ x})\ {\tt empty})\ {\tt s})
\end{align*}
or functions to quantify a predicate over a set:
\begin{align*}
  &{\tt exists},{\tt forall} : (\tau\to {\tt Formula})\to {\tt Set}\ \tau \to {\tt Formula} \\
  &{\tt exists\ p\ s} = {\tt not}\ ({\tt isEmpty}\ ({\tt filter\ p\ s})) \\
  &{\tt forall\ p\ s} = {\tt isEmpty}\ ({\tt filter}\ (\lambda x.{\tt not}\ ({\tt p\ x}))\ {\tt s})
\end{align*}
and so on. Note that these functions are written in exactly the same way as they would be for finite sets and the standard ${\tt Data.Set}$ type. This is not an accident, and indeed the programmer can use the convenient set-theoretic intuition of NLambda primitives.
For example, one could conveniently construct various orbit-finite sets such as the set of all pairs of atoms:
\[
	{\tt atomPairs} = {\tt sum}\ ({\tt map}\ (\lambda{\tt x}.{\tt map}\ (\lambda{\tt y}.({\tt x},{\tt y}))\ {\tt atoms})\ {\tt atoms}),
\]
the set of all pairs of {\em distinct atoms}:
\[
	{\tt distPairs} = {\tt filter}\ (\lambda({\tt x},{\tt y}).{\tt not}({\tt eq}\ {\tt x\ y}))\ {\tt atomPairs}
\]
and so on.

It should be stressed that all these constructions terminate in finite time, even though they formally involve infinite sets. To achieve this, values of orbit-finite set types ${\tt Set}\ \tau$ are internally not represented as lists or trees of elements of type $\tau$. Instead, they are stored and manipulated symbolically, using first-order formulas over variables that range over atom values. For example, the value of {\tt distPairs} above is stored as the formal expression:
\[
	\{(a,b) \mid a,b\in\mathbb{A},\ a\neq b\}
\]
or, more specifically, as a triple:
\begin{itemize}
\item a pair $(a,b)$ of ``atom variables'',
\item a list $[a,b]$ of those atom variables that are bound in the expression (in this case, the expression contains no free variables),
\item a formula $a\neq b$ over atom variables.
\end{itemize}
All the primitives listed above, such as {\tt isEmpty}, {\tt map} and {\tt sum}, are implemented on this internal representation. In some cases, this involves checking the satisfiability of certain formulas over atoms. In the current implementation of NLambda, an external SMT solver Z3~\cite{z3} is used for that purpose. For example, to evaluate the expression
$
	{\tt isEmpty\ distPairs},
$
NLambda makes a system call to the SMT solver to check whether the formula $a\neq b$ is satisfiable in the first-order theory of equality and, after receiving the affirmative answer, returns the value {\tt False}.

For more details about the semantics and implementation of NLambda, see~\cite{nlambda-paper}. The library itself can be downloaded from~\cite{nlambda-code}.

\subsection{Implementation of \nLStar\ and \nNLStar}

Using NLambda we implemented the algorithms from Sections~\ref{sec:nangluin}~and~\ref{sec:nondet}.
We note that the internal representation is slightly different than the one discussed in Section~\ref{sec:nangluin}.
Instead of representing the table $(S,E)$ with actual representatives of orbits, the sets are represented logically as described above.
Furthermore the control flow of the algorithm is adapted to fit in the functional programming paradigm.
In particular, recursion is used instead of a while loop.
In addition to the nominal adaptation of Angluin's algorithm $\nLStar$, we implemented a variant, $\nLStar_{col}$ which adds counterexamples to the columns instead of rows.

Target automata are defined using NLambda as well, using the automaton data type provided by the library.
Membership queries are already implemented by the library.
Equivalence queries are implemented by constructing a bisimulation (recall that bisimulation implies language equivalence),
where a counterexample is obtained when two DFAs are not bisimilar.
For nominal NFAs, however, we cannot implement a complete equivalence query as their language equivalence is undecidable.
We approximated the equivalence by bounding the depth of the bisimulation for nominal NFAs.
As an optimization, we use bisimulation up to congruence \cite{BonchiPous15}. Having an approximate teacher is a minor issue since in many applications no complete teacher can be implemented and one relies on testing \cite{Tomte15,BolligHLM13}.
For the experiments listed here the bound was chosen large enough for the learner to terminate with the correct automaton.

We remark that our algorithms seamlessly merge with teachers written in NLambda, but the current version of the library does not allow generating concrete membership queries for external teachers. We are currently working on a new version of the library in which this will be possible.

\subsection{Test cases}\label{sec:tests}

To provide a benchmark for future improvements, we tested our algorithms on a few simple automata described below. We report results in Table~\ref{tab:results}.
The experiments were performed on a machine with an Intel Core i5 (Skylake, 2.4 GHz) and 8 GB RAM.
\begin{table}
	\centering
	\small
	\begin{tabular}{l|rlS[table-format=3.2]S[table-format=3.2]rlS[table-format=3.2]}
		& \multicolumn{2}{c}{DFA} & {$\nLStar$ (s)} & {$\nLStar_{col}$ (s)} & \multicolumn{2}{c}{RFSA} & {$\nNLStar$ (s)} \\
		\hline
		$FIFO_0$   &   2 & 0 & 1.9  & 1.9  & 2  & 0 & 2.4 \\
		$FIFO_1$   &   3 & 1 & 12.9 & 7.4  & 3  & 1 & 17.3 \\
		$FIFO_2$   &   5 & 2 & 45.6 & 22.6 & 5  & 2 & 70.3 \\
		$FIFO_3$   &  10 & 3 & 189  & 107  & 10 & 3 & 476 \\
		$FIFO_4$   &  25 & 4 & 370  & 267  & 25 & 4 & 1230 \\
		$FIFO_5$   &  77 & 5 & 1337 & 697  & $\infty$  & $\infty$  & $\infty$ \\
		\hline
		$\lang_0$  &   2 & 0 & 1.3  & 1.4  & 2 & 0 & 1.4 \\
		$\lang_1$  &   4 & 1 & 29.6 & 4.7  & 4 & 1 & 8.9 \\
		$\lang_2$  &   7 & 2 & 229  & 23.1 & 7 & 2 & 84.7 \\
		\hline
		$\lang'_0$ &   3 & 1 & 4.4  & 4.9  & 3 & 1 & 11.3 \\
		$\lang'_1$ &   5 & 1 & 15.4 & 15.4 & 4 & 1 & 66.4 \\
		$\lang'_2$ &   9 & 1 & 46.3 & 40.5 & 5 & 1 & 210 \\
		$\lang'_3$ &  17 & 1 & 89.0 & 66.8 & 6 & 1 & 566 \\
		\hline
		$\lang_{eq}$ & n/a & n/a & n/a & n/a & 3 & 1 & 16.3 
	\end{tabular}
	\caption{Results of experiments.
		The column DFA (resp.\ RFSA) shows the number of orbits (left sub-column) and dimension (right sub-column) of the learnt minimal DFA (resp.\ canonical RFSA). We use $\infty$ when the running time is too high.}
	\label{tab:results}
\end{table}
\paragraph{Queue data structure.}
A queue is a data structure to store elements which can later be retrieved in a first-in, first-out order.
It has two operations: \texttt{push} and \texttt{pop}.
We define the alphabet $\Sigma_{FIFO} = \{ push(a), pop(a) \mid a \in \EAlph \}$.
The language $FIFO_n$ contains all valid traces of push and pop using a bounded queue of size $n$. The minimal nominal DFA for $FIFO_2$ is
\begin{center}
	\begin{automaton}
		\node[initial,accepting,state] (q0) {$q_0$};
		\node[state,accepting,right=55pt of q0] (q1) {$q_{1,x}$};
		\node[state,accepting,right=55pt of q1] (q2) {$q_{2, x, y}$};
		\node[state,below of=q0] (bot) {$\bot$};
		\path
		(q0) edge[bend left=15] node[trlab,above] {$push(x)$} (q1)
		(q1) edge[bend left=15] node[trlab,below] {$pop(x)$} (q0)
		(q1) edge[bend left=15] node[trlab,above] {$push(y)$} (q2)
		(q2) edge[bend left=15] node[trlab,below,xshift=-5pt] {$pop(x)$ to $q_{1,y}$} (q1)
		(q0) edge node[trlab,left] {$pop(\EAlph)$} (bot)
		(q1) edge[bend left] node[trlab,right] {$pop(\neq x)$} (bot)
		(q2) edge[bend left] node[trlab,right] {$pop(\neq x) / push(\EAlph)$} (bot)
		(bot) edge[loop left] node[trlab] {$\star$} (bot);
	\end{automaton}
\end{center}
The state reached from $q_{1,x}$ via $\xrightarrow{push(x)}$ is omitted: its outgoing transitions are those of $q_{2,x,y}$, where $y$ is replaced by $x$. Similar benchmarks appear in \cite{Tomte15,IsbernerHS14}.
\paragraph{Double word.}
$\lang_n = \{ ww \mid w \in \EAlph^n \}$ from Section~\ref{sec:overview}.
\paragraph{NFA.}
Consider the language $\lang_{eq} = \bigcup_{a \in \EAlph} \EAlph^\star a \EAlph^\star a \EAlph^\star$ of words where some letter appears twice.
This is accepted by an NFA which guesses the position of the first occurrence of a repeated letter $a$ and then waits for the second $a$ to appear.
The language is not accepted by a DFA \cite{BojanczykKL14}.
Despite this $\nNLStar$ is able to learn the automaton:
\begin{center}
	\begin{automaton}
		\node[initial,state] (q0) {$q_0'$};
		\node[state,right=55pt of q0] (q1) {$q_{1,x}'$};
		\node[state,accepting,right=55pt of q1] (q2) {$q_{2}'$};
		\path
		(q0) edge[bend left=15] node[trlab,above] {$x$} (q1)
		(q1) edge[bend left=15] node[trlab,below] {$\EAlph$} (q0)
		(q1) edge[bend left=15] node[trlab,above] {$x$} (q2)
		(q2) edge[bend left=15] node[trlab,below,xshift=-5pt] {$\EAlph$ to any $q_{2,x}'$} (q1)
		(q2) edge[bend left=45] node[trlab,below] {$\EAlph$} (q0)
		(q0) edge[loop above] node[trlab,above] {$\EAlph$} (q0)
		(q1) edge[loop above] node[trlab,above] {$\EAlph$} (q1)
		(q1) edge[loop below] node[trlab,left] {$y$ to $q_{2,y}'$} (q1)
		(q2) edge[loop above] node[trlab,above] {$\EAlph$} (q2);
	\end{automaton}
\end{center}
where the transition from $q_2'$ to $q_{1,x}'$ is defined as $\delta(q_2', a) = \{ q_{1,b}' \mid b \in \EAlph \}$.
\paragraph{$n$-last position.}
A prototypical example of regular languages which are accepted by very small NFAs is the set of words where a distinguished symbol $a$ appears on the $n$-last position \cite{BolligHKL09}.
We define a similar nominal language $\lang'_n = \bigcup_{a \in \EAlph} a \EAlph^\star a \EAlph^n$.
To accept such words non-deterministically, one simply guesses the $n$-last position.
This language is also accepted by a much larger deterministic automaton.
\section{Related work}\label{sec:related}

This section compares \nLStar\ with other algorithms from the literature. We stress that no comparison is possible for \nNLStar, as it is the first learning algorithm for non-deterministic automata over infinite alphabets.

The first one to consider learning automata over infinite alphabets was Sakamoto \cite{Sakamoto97}.
In his work the problem is reduced to $\LStar$ with some finite sub-alphabet.
The sub-alphabet grows in stages and $\LStar$ is rerun at every stage, until the alphabet is big enough to capture the whole language. In Sakamoto's approach, any learning algorithm can be used as a back-end.
This, however, comes at a cost: it has to be rerun at every stage, and each symbol is treated in isolation, which might require more queries.
Our algorithm \nLStar, instead, works with the whole alphabet from the very start, and it exploits its symmetry. An example is in Sections~\ref{sec:execution_example_original}~and~\ref{ssec:nom-learning}: the ordinary learner uses four equivalence queries, whereas the nominal one, using the symmetry, only needs three. Moreover, our algorithm is easier to generalize to other alphabets and computational models, such as non-determinism.

More recently papers appeared on learning register automata \cite{HowarSJC12,Cassel16}. Their register automata are as expressive as our deterministic nominal automata. The state-space is similar to our orbit-wise representation: it is formed by finitely many locations with registers. Transitions are defined symbolically using propositional logic. We remark that the most recent paper~\cite{Cassel16} generalizes the algorithm to alphabets with different structures (which correspond to different atom symmetries in our work), but at the cost of changing Angluin's framework. Instead of membership queries the algorithm requires more sophisticated tree queries. In our approach, using a different symmetry does not affect neither the algorithm nor its correctness proof. Tree queries can be reduced to membership queries by enumerating all $n$-types for some $n$ ($n$-types in logic correspond to orbits in the set of $n$-tuples). Keeping that in mind, their complexity results are roughly the same as ours, although this is hard to verify, as they do not give bounds on the length of individual tree queries. Finally, our approach lends itself better to be extended to other variations on $\LStar$ (of which many exist), as it is closer to Angluin's original work.

Another class of learning algorithms for systems with large alphabets is based on abstraction and refinement, which is orthogonal to the approach in the present paper but connections and possible transference of techniques are worth exploring in the future. In \cite{Tomte15}, the alphabet is reduced to a finite alphabet of abstractions, and $\LStar$ for ordinary DFAs over such finite alphabet is used. Abstractions are refined by counterexamples. Other similar approaches are \cite{HowarSM11,IsbernerHS13}, where global and local per-state abstractions of the alphabet are used, and \cite{MalerM14,MensM15}, where the alphabet can also have additional structure (e.g., an ordering relation). We can also mention \cite{BotincanB13}, a framework for learning symbolic models of software behavior.

In \cite{BergJR06,BergJR08}, authors cope with an infinite alphabet by running $\LStar$ (adapted to Mealy machines) using a finite approximation of the alphabet, which may be augmented when equivalence queries are answered. A smaller symbolic model is derived subsequently. Their approach, unlike ours, does not exploit the symmetry over the full alphabet. The symmetry allows our algorithm to reduce queries and to produce the smallest possible automaton at every step.

Finally we compare with results on session automata \cite{BolligHLM13}.
Session automata are defined over finite alphabets just like the work by Sakamoto.
However, session automata are more restrictive than deterministic nominal automata. For example, the model cannot capture an acceptor for the language of words where consecutive data values are distinct. This language can be accepted by a three orbit nominal DFA, which can be learned by our algorithm.

We implemented our algorithms in the nominal library NLambda as sketched before. Other implementation options include Fresh O'Caml~\cite{freshml}, a functional programming language designed for programming over nominal data structures with binding, and LOIS~\cite{lois,lois-popl}, a C++ library for imperative nominal programming. We chose NLambda for its convenient set-theoretic primitives, but the other options remain to be explored, in particular the low-level LOIS could be expected to provide more efficient implementations.
\section{Discussion and future work}

In this paper we defined and implemented extensions of several versions of \LStar and of \NLStar\ for nominal automata. 

We highlight two features of our approach:
\begin{itemize}
	\item it has strong theoretical foundations: the \emph{theory of nominal languages}, covering different alphabets and symmetries (see Section~\ref{sec:other-symms}); \emph{category theory}, where nominal automata have been characterized as \emph{coalgebras}~\cite{KozenMP015,CianciaM10} and many properties and algorithms (e.g., minimization) have been studied at this abstract level.
	\item it follows a generic pattern for transporting computation models and algorithms from finite sets to nominal sets, which leads to simple correctness proofs.
\end{itemize}
These features pave the way to several extensions and improvements.

Future work includes a general version of \nNLStar, parametric in the notion of side-effect (an example is non-determinism). Different notions will yield models with different degree of succinctness w.r.t.\ deterministic automata. The key observation here is that many forms of non-determinism and other side effects can be captured via the categorical notion of \emph{monad}, i.e., an algebraic structure, on the state-space. Monads allow generalizing the notion of composed and prime state: a state is composed whenever it is obtained from other states via an algebraic operation. Our algorithm \nNLStar\ is based on the \emph{powerset} monad, representing classical non-determinism. We are currently investigating a \emph{substitution} monad, where the operation is ``applying a (possibly non-injective) substitution of atoms in the support''. A minimal automaton over this monad, akin to a RFSA, will have states that can generate all the states of the associated minimal DFA via a substitution, but cannot be generated by other states (they are prime). For instance, we can give an automaton over the substitution monad that recognizes $\lang_2$ from Section~\ref{sec:overview}:
\begin{center}
\begin{automaton}
\node[initial,state] (q0) {$q_0$};	
\node[state,right of=q0] (qxp) {$q_x$};	
\node[state,right of=qxp] (qxyp) {$q_{xy}$};	
\node[state,right of=qxyp] (qys) {$q_{y}$};
\node[state,accepting,right of=qys] (q1) {$q_1$};
\node[state,right of=q1] (q2) {$q_2$};	
\path 
(q0) edge node[trlab,above] {$x$} (qxp)
(qxp) edge[bend right] node[trlab,above] {$y$} (qxyp)
(qxp) edge[bend left] node[trlab,above] {$x,[y \mapsto x]$} (qxyp)
(qxyp) edge node[trlab,above] {$x$} (qys)
(qxyp) edge[bend left,out=45,in=135] node[trlab,above] {$\neq x$} (q2)
(qys) edge node[trlab,above] {$y$} (q1)
(qys) edge[bend left] node[trlab,above] {$\neq y$} (q2)
(q1) edge node[trlab,above] {$A$} (q2)
(q2) edge[loop right] node[trlab,right] {$A$} (q2)
;
\end{automaton}	
\end{center}
Here $[y \mapsto x]$ means that, if that transition is taken, $q_{xy}$ (hence its language) is subject to $y \mapsto x$. In general, the size of the minimal DFA for $\lang_n$ grows more than exponentially with $n$, but an automaton with substitutions on transitions, like the one above, only needs $O(n)$ states.

In principle, thanks to the generic approach we have taken, all our algorithms should work for various kinds of atoms with more structure than just equality, as advocated in~\cite{BojanczykKL14}. Details, such as precise assumptions on the underlying structure of atoms necessary for proofs to go through, remain to be checked. For an implementation of automata learning over other kinds of atoms without compromising the generic approach, an extension of NLambda to those atoms will be needed, as the current version of the library only supports equality and totally ordered atoms.

The efficiency of our current implementation, as measured in Section~\ref{sec:tests}, leaves much to be desired. There is plenty of potential for running time optimization, ranging from improvements in the learning algorithms itself, to optimizations in the NLambda library (such as replacing the external and general-purpose SMT solver with a purpose-built, internal one, or a tighter integration of nominal mechanisms with the underlying Haskell language as it was done in~\cite{freshml}), to giving up the functional programming paradigm for an imperative language such as LOIS~\cite{lois,lois-popl}.

\paragraph{Acknowledgements.} We thank Frits Vaandrager and Gerco van Heerdt for useful comments and discussions. We also thank the anonymous reviewers.

\bibliographystyle{plainnat}
\bibliography{biblio}

\appendix
\section{Omitted proofs}

\subsection{Section \ref{sec:nangluin}}
\begin{lemma*}[{\bf \ref{lem:row_equiv}}]
	The relation $\sim$ is an equivariant equivalence relation.
\end{lemma*}
\begin{proof}
	Since the language $\lang$ is equivariant, so is $row$.
	The relation $\sim$ is defined as kernel of $row$ and is hence equivariant and an equivalence relation.
\end{proof}

\begin{lemma*}[{\bf \ref{lem:row_equiv}}]
	For all $u, v \in S$ we have that $u \equiv_\lang v$ implies $u \sim v$.
\end{lemma*}
\begin{proof}
	If $u \equiv_\lang$ then $\lang(uw) = \lang(vw)$ for all $w \in A^\star$.
	In particular $row(u)(e) = \lang(ue) = \lang(ve) = row(v)(e)$ for all $e \in E$.
	Hence $row(u) = row(v)$.
\end{proof}

\begin{lemma*}[{\bf \ref{lem:nerode_support}}]
	We have $\supp([u]_{\sim}) \subseteq \supp([u]_{\equiv_\lang}) \subseteq \supp(u)$ for all $u \in S$.
\end{lemma*}
\begin{proof}
	By Lemma~\ref{lem:nerode_approx} we have an equivariant function $S / {\equiv_\lang} \to S / {\sim}$ sending $[u]_{\equiv_\lang}$ to $[u]_{\sim}$.
	Since equivariant functions preserve supports (Lemma~4.8 in \cite{BojanczykKL14}) we have $\supp([u]_{\sim}) \subseteq \supp([u]_{\equiv_\lang})$.
\end{proof}

\subsection{Section \ref{sec:nondet}}
\begin{lemma*}[{\bf \ref{lem:can-rfsa}}]
	The canonical nominal RFSA is a proper nominal NFA.
\end{lemma*}
\begin{proof}
	\hfill
	\begin{itemize}
		\item $Q$ is orbit-finite, as $PR(\lang)$ is an equivariant subset of $R(\lang)$, which is orbit-finite. We now show equivariance of $PR(\lang)$. Suppose $\lang' \in PR(\lang)$ but $\pi \cdot \lang' \notin PR(\lang)$. We have $\pi \cdot \lang' \in R(\lang)$, by equivariance of $R(\lang)$, so $\pi \cdot \lang'$ is the union of the finite residuals it strictly contains. By equivariance of $R(\lang)$, $\lang'$ is composed: its components are those of $\pi \cdot \lang'$ permuted via $\pi^{-1}$. This is a contradiction.
		\item $Q_0$ is an equivariant, orbit-finite subset of $PR(\lang)$. We have $\supp(\lang) = \emptyset$ and $\lang = \Union Q_0$ so $\supp(Q_0) = \emptyset$, i.e., $Q_0$ is equivariant. Being an an equivariant subset of $Q$, which is orbit-fine, it is itself orbit-finite.
		\item $F$ is equivariant. In fact, take $\lang' \in F$. By equivariance of $Q$, $\pi \cdot \lang' \in Q$, and we have $\epsilon = \pi \cdot \epsilon \in \pi \cdot \lang'$, so $\pi \cdot \lang' \in F$. Orbit-finiteness follows from the usual argument.
		\item $\delta$ is an equivariant relation, i.e. $(\lang_1,a,\lang_2) \in \delta$ implies $(\pi \cdot \lang_1, \pi \cdot a , \pi \cdot \lang_2) \in \delta$. Suppose $\lang_2 \subseteq \lder{a}\lang_1$, i.e., $v \in \lang_2$ implies $av \in \lang_1$. Take $w \in \pi \cdot \lang_2$. Then $\pi^{-1} \cdot w \in \lang_2$, so $a (\pi^{-1} w) \in \lang_1$, and then $(\pi \cdot a)w \in \pi \cdot \lang_1$, as required.
	\end{itemize}
\end{proof}

\begin{lemma*}[{\bf \ref{lem:pr-ns}}]
	$PR(S,E)$, $PR^\top(S,E)$ and $PR(S,E) \setminus PR^\top(S,E)$ are orbit-finite nominal sets.
\end{lemma*}
\begin{proof}
	We prove that $PR(S,E)$ is equivariant. We know that the set $R(S,E)$ of rows of $(S,E)$ is an orbit-finite nominal set. Suppose by contradiction that $row(t) \in PR(S,E)$ but $\pi \cdot row(t) \in R(S,E) \setminus PR(S,E)$, so
	\[
	\pi \cdot row(t) = \Rowunion \{ row(s) \mid row(s) \rowincls \pi \cdot row(t) \} \enspace .
	\]
	Then, by equivariance of $\rowincls$, $row(t)$ is the union of $\pi^{-1} \cdot row(s)$, with $row(s)$ in the equation above, so it is composed, a contradiction. Since $PR(S,E)$ is an equivariant subset of an orbit-finite nominal set, it is orbit-finite.
	A similar argument gives $PR^\top(S,E)$ orbit-finite. Finally, $PR(S,E) \setminus PR^\top(S,E)$ is obtained by removing some orbits from $PR(S,E)$, so it is orbit-finite.
\end{proof}

\begin{lemma*}[{\bf \ref{lem:comp-prop}}]
The set $C(row(t))$ has the following properties:
\begin{enumerate}
	\item $\supp(C(row(t))) \subseteq \supp(row(t))$. 
	\item it is orbit-finite under the action of the group
	\[
		G_t = \{ \pi \in Perm(\EAlph) \mid \restr{\pi}{\supp(row(t))} = id \}
	\]
	of permutations fixing $\supp(row(t))$.
\end{enumerate}
\end{lemma*}
\begin{proof}
\hfill

\begin{enumerate}
	\item 
We prove that permutations fixing $\supp(row(t))$ also fix $C(row(t))$. Take $\pi$ such that $\restr{\pi}{\supp(row(t))} = id$ and $row(s) \in C(row(t))$. We want to prove that $\pi \cdot row(s) \in C(row(t))$. We have $row(s) \rowincl row(t)$ and, by equivariance of $\rowincl$, $\pi \cdot row(s) \rowincl \pi \cdot row(t)$. But $\pi \cdot row(t) = row(t)$, as $\pi$ fixes $\supp(row(t))$, so $\pi \cdot row(s) \rowincl row(t)$, as required.

\item
In the previous point we proved that $C(row(t))$ is equivariant w.r.t.\ permutations in $G_t$. To show that it is orbit-finite, it is enough to show that $PR^{\top}(S,E)$ is orbit-finite w.r.t.\ $G_t$; orbit-finiteness of $C(row(t))$ will follow by equivariance. 

We call $G_t$-orbit (resp.\ $Perm(\EAlph)$-orbit) an orbit w.r.t.\ the group of permutations $G_t$ (resp.\ $Perm(\EAlph)$).
Each $G_t$-orbit of $PR^{\top}(S,E)$ is clearly included in one of its ($Perm(\EAlph)$-)orbits, so we show that each orbit of the latter type contains finitely-many orbits of the former type. Every $Perm(\EAlph)$-orbit of dimension $k$ in $PR^{\top}(S,E)$ is image of $\EAlph^{(k)}$ via an equivariant surjective function. Such functions can only decrease the number of orbits, so we prove that $\EAlph^{(k)}$ is orbit-finite w.r.t.\ $G_t$. Take $(a_1,\dots,a_k) \in \EAlph^{(k)}$. Then $(a_1,\dots,a_k)$ and $(a_1',\dots,a_k')$ are in the same $G_t$-orbit if and only if $a_i \in \supp(row(t))$ implies $a_i = a'_i$. Hence, the number of $G_t$-orbits of $PR^{\top}(S,E)$ contained in one of its $Perm(\EAlph)$-orbit is at most the number of subsets of size $k$ of a set of size $|\supp(row(t))|$.
\end{enumerate}
\end{proof}

\begin{lemma*}[{\bf \ref{lem:nlstarwelldef}}]
	The conjecture automaton $N(S,E)$ (see Definition~\ref{def:nfa-conj}) is a nominal NFA.
\end{lemma*}
\begin{proof}
	Let $N(S,E) = (Q,Q_0,F,\delta)$. Then:
	\begin{itemize}
		\item $Q = PR(S,E)$ is an orbit-finite nominal set (Lemma~\ref{lem:pr-ns});
		\item $Q_0$ is an equivariant, thus orbit-finite, subset of $Q$. In fact $\supp(row(\epsilon)) \subseteq \supp(\epsilon) = \emptyset$ (Lemma~\ref{lem:nerode_support}), so $\supp(Q_0) = \emptyset$ by Lemma~\ref{lem:comp-prop}(\ref{supp}), i.e., it is equivariant.
		\item $F$ is an orbit-finite, thus orbit-finite, subset of $Q$. Suppose $r \in F$, i.e., $r(\epsilon) = 1$. Then
		\[
		(\pi \cdot r)(\epsilon) = \pi \cdot r( \pi^{-1} \cdot \epsilon) = \pi \cdot r(\epsilon) = 1
		\]
		where the second equation uses the action of permutations over  functions $2^E$. Therefore $(\pi \cdot r) \in F$. 
		\item The transition relation $\delta$ is equivariant. To prove this, we need to show that $(row(s),a,r') \in \delta$ implies $(\pi \cdot row(s),\pi(a),\pi \cdot r') \in \delta$. By definition, $(row(s),a,r') \in \delta$ only if $r' \rowincl row(sa)$. Using equivariance of $\rowincl$, $\pi \cdot r' \rowincl \pi \cdot row(sa)$ and, using equivariance of $row$, $\pi \cdot row(sa) = row((\pi \cdot s)\pi(a))$. Therefore $\pi \cdot r' \rowincl row((\pi \cdot s)\pi(a))$, which gives $(\pi \cdot row(s),\pi(a),\pi \cdot r') \in \delta$.
	\end{itemize}
\end{proof}

\end{document}